\documentclass[a4paper,11pt]{article}
\usepackage{cite}
\usepackage{hyperref}
\usepackage[top=1 in,bottom=1.0 in,left=1.0 in,right=1.0 in]{geometry}
\hypersetup{hidelinks}
\usepackage{amsthm}
\usepackage{amssymb,amsmath}

\numberwithin{equation}{section}

\newtheorem{proposition}{Proposition}

\newtheorem{theorem}{Theorem}
\newtheorem{remark}{Remark}

\newcommand{\bs}{\boldsymbol}
\newcommand{\mb}{\mathbb}
\newcommand{\mf}{\mathbf}
\newcommand{\mr}{\mathrm}
\newcommand{\mk}{\mathcal}

\newcommand{\bsb}{\begin{subequations}}
\newcommand{\esb}{\end{subequations}}

\newcommand{\ti}[1]{\tilde{#1}}

\usepackage{amsmath}
\usepackage{amsfonts}
\usepackage{latexsym}
\usepackage{amssymb}
\usepackage{appendix}
\usepackage{mathrsfs}
\usepackage{graphicx}
\usepackage{bm}
\usepackage{ulem}
\usepackage{arydshln}
\usepackage{tikz}
\usepackage{multirow}
\usepackage{rotating}

\usepackage{color}
\newcommand{\bred}{\begin{color}{red}}
\newcommand{\ecl}{\end{color}}
\newcommand{\bblue}{\begin{color}{blue}}
\newcommand{\bgre}{\begin{color}{green}}
\newcommand{\bora}{\begin{color}{orange}}
			
\begin{document}
				
\title{Direct linearization of the SU(2) anti-self-dual Yang-Mills equation in various spaces}
\author{Shangshuai Li$^{1,2,3}$\footnote{Email: lishangshuai@shu.edu.cn},
~~ Da-jun Zhang$^{1,2}$\footnote{Email: djzhang@staff.shu.edu.cn} \\
{\small  $^{1}$Department of Mathematics, Shanghai University, Shanghai 200444, P.R. China}\\
{\small $^{2}$Newtouch Center for Mathematics of Shanghai University,  Shanghai 200444, P.R. China}\\
{\small $^{3}$Department of Applied Mathematics, Faculty of Science and Engineering, Waseda University,}\\
{\small Tokyo 169-8555, Japan}
}

\maketitle
				
\begin{abstract}
The paper establishes a direct linearization scheme for the  SU(2) anti-self-dual Yang-Mills (ASDYM) equation.
The scheme starts from a set of  linear integral equations with general measures and plane wave factors.
After introducing infinite-dimensional matrices as master functions,
we are able to investigate evolution relations and recurrence relations of these functions,
which lead us to the unreduced ASDYM equation.
It is then reduced to the ASDYM equation in the  Euclidean space and two ultrahyperbolic
spaces by reductions to meet the reality conditions and gauge conditions, respectively.
Special solutions can be obtained by choosing suitable measures.

\begin{description}
\item[Keywords:] anti-self-dual Yang-Mills equation, direct linearization, linear integral equation, integrable system

\end{description}

\end{abstract}

\section{Introduction}\label{sec-1}

Yang-Mills theory provides one of the most beautiful models in quantum physics,
which describes the interactions between elementary particles.
Nowadays, Yang-Mills theory plays an important role in the field of nonlinear science, twistor theory,
duality structure, mathematical physics and so on.
An attractive feature among these researches is the self-dual condition in Yang-Mills field (gauge field),
which leads to an integrable system \cite{Jimbo-1982,Ward-1984},
namely, the self-dual Yang-Mills (SDYM) equation.

In the pioneering work on constructing solutions to the SDYM equation,
instantons were proposed as a special action that is globally finite.
Instantons can be constructed through the Atiyah-Hitchin-Drinfeld-Manin construction \cite{Atiyah-1978}
and the Corrigan-Fairlie-'t Hooft-Wilczek ansatz \cite{Corrigan-1977,Wilczek-1977}.
The original self-dual condition is difficult to treat directly, yet Yang's work \cite{Yang-1977}
proposed a coordinate transformation that allows the self-dual condition to be represented by a simpler form,
leading to a matrix equation \cite{BFNY-1978,Po-1980} (known as Yang's $J$-formulation):
\begin{equation}\label{Yang-SDYM}
	(J_yJ^{-1})_{\tilde y}+(J_zJ^{-1})_{\tilde z}=0,
\end{equation}
where $y,z$ are independent coordinates and $\tilde y=y^*, \tilde z=z^*$,
i.e.  the complex conjugates of $y$ and $z$, respectively.
For a given gauge group, such as GL$(N)$, SL$(N)$, U$(N)$ and SU$(N)$,
the gauge potentials take values from the Lie algebra of the gauge group.
This leads that the solutions of \eqref{Yang-SDYM} will satisfy a certain condition,
which we call the gauge condition for convenience (see section \ref{sec-2-2} for detailed interpretation).

In fact, most of discussions about the SDYM equation as well as Yang's $J$-formulation \eqref{Yang-SDYM}
focus on describing the motion in Euclidean space $\mb E$.
There are also several researches on other spaces,
such as the Minkowski space $\mb M$ \cite{deVega-1988,Getmanov-1990,Chau-1993,Chau-1994}
and the ultrahyperbolic space $\mb U$ \cite{Chau-1993,Chau-1994,Mason-2005}.
Recently, a further progress has been made on constructing soliton solutions of the SDYM equation in various spaces
by means of quasideterminant \cite{Hamanaka-2020,Hamanaka-2022,Huang-2021}.
A more recent work showed the application of SDYM equation on 4D Wess-Zumino-Witten model,
which describes the open $N=2$ string theory \cite{Hamanaka-2023}.

As an integrable system, the SDYM equation has been solved via many integrable methods
\cite{Hamanaka-2020,Hamanaka-2022,Belavin-1978,Ueno-1982,Takasaki-1984,Sasa-1998,Nimmo-2000}.
Our recent work employed a direct approach to derive the explicit solutions of
the SU(2) SDYM equation \cite{LQYZ-SAPM-2022}, which reveals the connection between the SDYM equation and
the Ablowitz-Kaup-Newell-Segur (AKNS) system from the perspective of the Cauchy matrix structure \cite{Zhao-2018}.
This approach was then extended  to deal with the SU($N$) case \cite{LSS-2023},
where a new link between the SDYM equation and the matrix Kadomtsev-Petviashvili (KP) hierarchy was revealed.

In this paper, we aim to build up the direct linearization scheme
the anti-self-dual Yang-Mills (ASDYM) equation.
Note that the SDYM equation and ASDYM equation are mathematically equivalent.
The direct linearization (DL) approach is an effective tool for studying integrable equations and
understanding the relationships between them by starting from different types of linear integral equations.
It was first proposed by Fokas and Ablowitz to examine the Korteweg-de Vries (KdV) equation \cite{Fokas-1981}.
Later, this approach was employed to construct solutions to the nonlinear Schr\"{o}dinger equation \cite{Nijhoff-1983}
and the derivative nonlinear Schr\"{o}dinger equation \cite{Nijhoff-1983-2}, etc.
Soon after, the approach was developed by Dutch group
to be an effective tool for investigating discrete integrable systems \cite{83-NQC,QNCL-84,85-NCW}.
Recent progresses about the DL approach include
the application to the extended discrete Boussinesq equations \cite{ZZN-SAM-2012}
and discrete KP type equations \cite{FW-2017,FW-thesis,FW-2020,FW-Nijhoff-2022}, and
the establishment of an elliptic version for elliptic solitons of discrete integrable systems \cite{NSZ-2023}.
Based on the link between the SDYM equation and the AKNS system that we found in \cite{LQYZ-SAPM-2022},
in this paper, we will establish the DL scheme for the negative and positive AKNS hietarchy
from which we are able to construct solutions of the ASDYM equation.

The paper is organized as follows.
We begin by introducing the formulation of the ASDYM equation and  gauge conditions in section \ref{sec-2}.
In section \ref{sec-3}, we establish the DL scheme to derive the unreduced ASDYM equation.
We introduce a set of linear integral equations,
define some infinite-dimensional matrices
and derive their evolution relations and ``difference'' relations.
Then we can extract relations for some  $2\times2$ matrices,
which lead us to the unreduced ASDYM equation.
Then in section \ref{sec-4}, we implement reductions to meet the reality conditions and gauge conditions
to obtain the SU(2) ASDYM equation in different spaces.
Finally, concluding remarks are given in section \ref{sec-5}.

\section{Formulation of the ASDYM equation}\label{sec-2}

In this section, we recall the $J$-formulation and gauge conditions of the ASDYM equation\footnote{
The ASDYM equation is mathematically equivalent to the SDYM equation
in the sense of a coordinate transformation: $(z_3, z_4) \rightarrow (z_4, z_3)$.
For convenience and without losing generality, we only consider the anti-self-dual condition in this section.}
in different spaces.
One can refer to \cite{MW-book} for more details.

\subsection{$J$-formulation}\label{sec-2-0}

The metric matrix in $\mb C_4=(z_1,z_2,z_3,z_4)$ is given by
\begin{equation}\label{eta-mn}
	(\eta^{mn})_{4\times4}\doteq
	\begin{pmatrix}
		0 & 1 & 0 & 0 \\
		1 & 0 & 0 & 0 \\
		0 & 0 & 0 & -1 \\
		0 & 0 & -1 & 0
	\end{pmatrix},
	~~~~m,n=1,2,3,4,
\end{equation}
which is called the double null coordinates.
The metric can be calculated as
\begin{equation}\label{metric}
	\mr ds^2=\eta^{mn}\mr dz_m\mr dz_n=2(\mr dz_1\mr dz_2-\mr dz_3\mr dz_4).
\end{equation}
Let $G$ be the gauge group of a certain gauge field and $g$ be the Lie algebra of the group $G$.
The field strengths in this gauge field are then defined by
\begin{equation}
	F_{ij}\doteq[\mathcal D_i,\mathcal D_j]=\partial_i A_j-\partial_j A_i+[A_i,A_j],~~~~
\mathcal D_i\doteq\partial_i+A_i,
\end{equation}
where $[\cdot, ~\cdot]$ is the Lie bracket defined as $[H,G] = HG-GH$, $A_i\in g$ for $i=1,2,3,4$ are gauge potentials,
operator $\mathcal D_i$ is the covariant derivative and $\partial_i=\partial_{z_i}$.
The anti-self-dual condition is introduced as
\begin{equation}\label{ASD-condition}
	F_{ij}=-\frac{1}{2}\epsilon_{ijkl}\eta^{ka}\eta^{lb}F_{ab},
\end{equation}
where $\epsilon_{ijkl}$ is the Levi-Civita tensor, $\eta^{mn}$ follows the definition in \eqref{eta-mn},
and $i,j,k,l,a,b$ are indexes running over $\{1,2,3,4\}$.
Hence one can expand \eqref{ASD-condition} to obtain the following equations:
\begin{equation}\label{ASD-condition-expand}
	F_{13}=0,~~~~F_{24}=0,~~~~F_{12}-F_{34}=0.
\end{equation}
Let $(z,\ti z,w,\ti w)\doteq(z_1,z_2,z_3,z_4)$, the first two equations of \eqref{ASD-condition-expand}
indicate that there exist two invertible matrices $h$ and $\ti h$ such that
\begin{equation}
	\mathcal D_zh=0,~~~~\mathcal D_wh=0,~~~~\mathcal D_{\ti z}\ti h=0,~~~~\mathcal D_{\ti w}\ti h=0.
\end{equation}
$h,\ti h$ are called the generating matrices since the gauge potentials can be represented as
\begin{equation}\label{A-by-h}
	A_z=-(\partial_zh)h^{-1},~~~~A_w=-(\partial_wh)h^{-1},~~~~
A_{\ti z}=-(\partial_{\ti z}\ti h)\ti h^{-1},~~~~A_{\ti w}=-(\partial_{\ti w}\ti h)\ti h^{-1}.
\end{equation}
Then, by defining $J=h^{-1}\ti h$, the $J$-formulation of ASDYM (as well as SDYM) equation
can be derived from \eqref{ASD-condition-expand} as
\begin{equation}\label{ASDYM-unreduce}
	\partial_{\ti z}(J^{-1}\partial_zJ)-\partial_{\ti w}(J^{-1}\partial_wJ)=0.
\end{equation}
We call it the unreduced ASDYM equation.

\subsection{The ASDYM equation in different spaces}\label{sec-2-1}

Notice that the unreduced ASDYM equation \eqref{ASDYM-unreduce} reduces to \eqref{Yang-SDYM}
under the reality condition $\ti z=z^*$ and $\ti w=-w^*$,
where $*$ denote the complex conjugate.
This reality condition can be achieved through the following coordinate transformation:
\begin{equation}\label{coor-E}
	\begin{pmatrix}
		\ti z & w \\
		\ti w & z
	\end{pmatrix}
	=\frac{\sqrt2}{2}
	\begin{pmatrix}
		x^0+\mr ix^1 & -x^2+\mr ix^3 \\
		x^2+\mr ix^3 & x^0-\mr ix^1
	\end{pmatrix},
\end{equation}
where $\mathrm{i}^2=-1$, $x^0,x^1,x^2,x^3$ are real coordinates, i.e.
$(x^0,x^1,x^2,x^3)\in \mathbb{R}^4$.
The metric \eqref{metric} under the coordinate transformation \eqref{coor-E}
satisfies the Euclidean signature $(+,+,+,+)$, i.e.
\begin{align*}
	\mr d s^2=2(\mr dz\mr d\ti z-\mr dw\mr d\ti w)&
=(\mr dx^0+\mr i\mr dx^1)(\mr dx^0-\mr i\mr dx^1)+(\mr dx^2+\mr i\mr dx^3)(\mr dx^2-\mr i\mr dx^3) \\
	&=(\mr dx^0)^2+(\mr dx^1)^2+(\mr dx^2)^2+(\mr dx^3)^2.
\end{align*}
For the Minkowski space $\mb M$ with the signature $(+,-,-,-)$, the coordinate transformation is given by
\begin{equation}
	\begin{pmatrix}
		\ti z & w \\
		\ti w & z
	\end{pmatrix}
	=\frac{\sqrt2}{2}
	\begin{pmatrix}
		x^0+x^1 & x^2-\mr ix^3 \\
		x^2+\mr ix^3 & x^0-x^1
	\end{pmatrix},
\end{equation}
and the reality condition is $z,\ti z\in\mb R$, $\ti w=w^*$.
As for the ultrahyperbolic space $\mb U$ with split signature $(+,+,-,-)$,
there can be two different coordinate transformations:
\bsb\label{coor-U}
\begin{align}
	\label{coor-U-1}\begin{pmatrix}
		\ti z & w \\
		\ti w & z
	\end{pmatrix}
	&=\frac{\sqrt2}{2}
	\begin{pmatrix}
		x^0+\mr ix^1 & x^2-\mr ix^3 \\
		x^2+\mr ix^3 & x^0-\mr ix^1
	\end{pmatrix}, \\
	\label{coor-U-2}\begin{pmatrix}
		\ti z & w \\
		\ti w & z
	\end{pmatrix}
	&=\frac{\sqrt2}{2}
	\begin{pmatrix}
		x^0+x^2 & x^3-x^1 \\
		x^3+x^1 & x^0-x^2
	\end{pmatrix}.
\end{align}
\esb
Both cases result in the same metric, i.e.,
\begin{align*}
	\mr d s^2=2(\mr dz\mr d\ti z-\mr dw\mr d\ti w)&
=(\mr dx^0+\mr i\mr dx^1)(\mr dx^0-\mr i\mr dx^1)-(\mr dx^2+\mr i\mr dx^3)(\mr dx^2-\mr i\mr dx^3) \\
	&=(\mr dx^0)^2+(\mr dx^1)^2-(\mr dx^2)^2-(\mr dx^3)^2,
\end{align*}
and
\begin{align*}
	\mr d s^2=2(\mr dz\mr d\ti z-\mr dw\mr d\ti w)
&=(\mr dx^0-\mr dx^2)(\mr dx^0+\mr dx^2)-(\mr dx^3-\mr dx^1)(\mr dx^3+\mr dx^1) \\
	&=(\mr dx^0)^2+(\mr dx^1)^2-(\mr dx^2)^2-(\mr dx^3)^2.
\end{align*}
Thus, the reality condition of \eqref{coor-U-1} is  $\ti z=z^*$, $\ti w=w^*$.
For \eqref{coor-U-2}, it is $z,\ti z,w,\ti w\in\mb R$.

For convenience, we denote $\mb U$ with coordinate \eqref{coor-U-1} as $\mb U_1$
and denote $\mb U$ with coordinate \eqref{coor-U-2} as $\mb U_2$.
Then, the above results can be summarized in Table \ref{Tab-1}.
\begin{table}[ht]
	\begin{center}
		\begin{tabular}{|c|c|c|} \hline
			Space  & Reality condition & ASDYM equation   \\ \hline
			$\mb E$ & $\ti z=z^*,\ti w=-w^*$
& $\partial_{z^*}(J^{-1}\partial_zJ)+\partial_{w^*}(J^{-1}\partial_wJ)=0$   \\ \hline
			$\mb M$ & $z,\ti z\in\mb R,\ti w=w^*$	
&	$\partial_{\ti z}(J^{-1}\partial_zJ)-\partial_{w^*}(J^{-1}\partial_wJ)=0$  \\ \hline
			$\mb U_1$ & $\ti z=z^*$, $\ti w=w^*$	
&	$\partial_{z^*}(J^{-1}\partial_zJ)-\partial_{w^*}(J^{-1}\partial_wJ)=0$  \\ \hline
			$\mb U_2$ & $z,\ti z,w,\ti w\in\mb R$	
&	$\partial_{\ti z}(J^{-1}\partial_zJ)-\partial_{\ti w}(J^{-1}\partial_wJ)=0$  \\ \hline
		\end{tabular}
	\end{center}
	\caption{The ASDYM equations in various spaces}\label{Tab-1}
\end{table}

\subsection{Gauge conditions}\label{sec-2-2}

The gauge condition of Yang's $J$-formulation \eqref{Yang-SDYM} in SU($N$) group requires
$J$ to be a Hermitian matrix  with determinant being one \cite{Po-1980,LQYZ-SAPM-2022,LSS-2023}.
By mentioning the concept of gauge condition,
we mean certain conditions satisfied by matrix $J$
when the gauge potentials take values from the Lie algebra of a certain gauge group.

In the following, we introduce the $\mr{SU}(N)$ gauge conditions of the ASDYM equation in various spaces.
To achieve that, we first consider gauge conditions for the gauge group $\mr{SL}(N)$ and $\mr{U}(N)$ separately,
and then combine them together.

For  the case of $G=\mr{SL}(N)$, the gauge potentials satisfy
\begin{equation}
	\mr{tr}(A_{\alpha})=0,~~~~\mr{tr}(A_{\ti\alpha})=0,~~~~\alpha=z,w,
\end{equation}
which is valid no matter the metric space is Euclidean, ultrahyperbolic or of Minkovski.
This indicates
\begin{equation}
	\mr{tr}((\partial_\alpha h)h^{-1})=\frac{\partial_\alpha|h|}{|h|}=0,~~~~
	\mr{tr}((\partial_{\ti \alpha}\ti h)\ti h^{-1})=\frac{\partial_{\ti \alpha}|\ti h|}{|\ti h|}=0.
\end{equation}
which means both $|h|$ and $|\ti h|$ are constants.
We take $|h|=|\ti h|=1$ by normalization, which does not lose  generality.
Thus the determinant of $J$ can be calculated as $|J|=|h^{-1}\ti h|=|h^{-1}||\ti h|=1$.

For the case of $G=\mr{U}(N)$, the gauge conditions of $J$ are different in various spaces.
So we need to introduce them separately.
Starting from the Euclidean space $\mb E$ where the reality condition is $\ti z=z^*,\ti w=-w^*$, we determine
\begin{equation}\label{new-gauge-potential}
	A_z\doteq\frac{\sqrt 2}{2}(A_0+\mr iA_1),~~~A_{z^*}\doteq\frac{\sqrt 2}{2}(A_0-\mr iA_1),~~~
	A_w\doteq\frac{\sqrt 2}{2}(A_2+\mr iA_3),~~~A_{w^*}\doteq\frac{\sqrt 2}{2}(A_2-\mr iA_3).
\end{equation}
The gauge potentials satisfy
\begin{equation}\label{gauge-condition}
	(A_i)^\dagger=-A_i,~~~~i=0,1,2,3,
\end{equation}
where $A_i^\dagger=(A_i^*)^T$. Then we have
\begin{equation}\label{new-gauge-condition}
	(A_z)^\dagger=-A_{z^*},~~~~(A_w)^\dagger=-A_{w^*}.
\end{equation}
By using \eqref{A-by-h}, one obtains
\bsb\label{relation-h-E}
\begin{align}
	\label{potential-w}&(h^{-1})^\dagger(\partial_{z^*}h^\dagger)=-(\partial_{z^*}\ti h)\ti h^{-1}
	\rightarrow
	\partial_{z^*}(\ti h^{-1})=(h^\dagger\ti h)^{-1}(\partial_{z^*}h^\dagger), \\
	\label{potential-ti-w}&(h^{-1})^\dagger(\partial_{w^*}h^\dagger)=-(\partial_{w^*}\ti h)\ti h^{-1}
	\rightarrow
	\partial_{w^*}(\ti h^{-1})=(h^\dagger\ti h)^{-1}(\partial_{w^*}h^\dagger).
\end{align}
\esb
Relations \eqref{potential-w} and \eqref{potential-ti-w} are compatible with the condition
$h=(\ti h^\dagger)^{-1}$. Thus $J=h^{-1}\ti h=\ti h^\dagger\ti h$, which is a Hermitian matrix.
Next, for the space $\mb U_1$, where $\ti z=z^*, \ti w=w^*$,
the definitions of gauge potentials still follow \eqref{new-gauge-potential} and
they satisfy \eqref{new-gauge-condition}, which reveals the same relations of $h$ and $\ti h$ as in \eqref{relation-h-E}.
Thus in this case, the solution $J$ will also be a Hermitian matrix.
Finally, for the  $\mb U_2$ space, we determine
\begin{equation}
	A_z=\frac{\sqrt2}{2}(A_0-A_2),~~~A_{\ti z}=\frac{\sqrt2}{2}(A_0+A_2),~~~
	A_w=\frac{\sqrt2}{2}(A_3-A_1),~~~A_{\ti w}=\frac{\sqrt2}{2}(A_3+A_1).
\end{equation}
According to \eqref{gauge-condition}, we soon obtain
\begin{equation}
	(A_{\alpha})^\dagger=-A_{\alpha},~~~~(A_{\ti\alpha})^\dagger=-A_{\ti\alpha},~~~~a=z,w.
\end{equation}
Thus
\begin{equation}
	\partial_\alpha(h^{-1})=(h^\dagger h)^{-1}(\partial_{\alpha}h^\dagger),~~~~
	\partial_{\ti\alpha}(\ti h^{-1})=(\ti h^\dagger \ti h)^{-1}(\partial_{\ti \alpha}\ti h^\dagger),
\end{equation}
which holds when $h^\dagger h=\ti h^\dagger \ti h=\bs I_N$, where $\bs I_N$ stands for the $N$th-order identity matrix.
Hence we have the solution constructed by $J=h^{-1}\ti h$ that   satisfies the unitary condition:
\begin{align}
	J^\dagger J=(h^{-1}\ti h)^\dagger(h^{-1}\ti h)=\ti h^\dagger(h^{-1})^\dagger h^{-1}\ti h=\bs I_N.
\end{align}
Note that there is no such nontrivial gauge condition for the case of the Minkowski space  (see Remark \ref{rem-1}).

Especially, when $G$ is of $\mr{SU}(N)$, solution $J$ satisfies the gauge condition of $\mr{SL}(N)$
and $\mr{U}(N)$ simultaneously.
We summarize the above results in the following table.

\begin{table}[ht]
	\begin{center}
		\begin{tabular}{|c|c|} \hline
			Spaces  & Gauge condition   \\ \hline
			$\mb E$ & $|J|=1$, $J^\dagger=J$ \\ \hline
			$\mb U_1$ & $|J|=1$,  $J^\dagger=J$ \\ \hline
			$\mb U_2$ & $|J|=1$, $J^\dagger J=\bs I_N$  \\ \hline
		\end{tabular}
	\end{center}
	\caption{Gauge conditions of SU($N$) gauge groups in different spaces}\label{Tab-2}
\end{table}

\begin{remark}\label{rem-1}
	In the Minkowski space $\mb M$, there is no such nontrivial gauge condition of $J$
to make the gauge potentials belong to $\mr{su}(N)$.
In this case, the gauge potentials are determined as
	\begin{equation}
		A_z=\frac{\sqrt2}{2}(A_0-A_1),~~~A_{\ti z}=\frac{\sqrt2}{2}(A_0+A_1),~~~
		A_w=\frac{\sqrt2}{2}(A_2+\mr iA_3),~~~A_{\bar w}=\frac{\sqrt2}{2}(A_2-\mr iA_3).
	\end{equation}
	Then the generating matrices $h,\ti h$ should satisfy
	\begin{align}
		\partial_z(h^{-1})=(h^\dagger h)^{-1}(\partial_{z}h^\dagger),~~~~
		\partial_{\ti z}(\ti h^{-1})=(\ti h^\dagger \ti h)^{-1}(\partial_{\ti z}\ti h^\dagger),~~~~
		\partial_{\bar w}(\ti h^{-1})=(h^\dagger\ti h)^{-1}(\partial_{\bar w}h^\dagger),
	\end{align}
	which requires that $h^\dagger h=\ti h^\dagger\ti h=\bs I_N$ and $\ti h^{-1}=h^\dagger$ hold simultaneously.
Hence $J$ will be a Hermitian matrix and a unitary matrix, which shows $J^\dagger J=J^2=\bs I_N$ and
thus $J$ has to be $J=\mr{diag}(j_1,j_2,\cdots,j_N)$, $j_i=\pm1$, $i=1,\cdots,N$.
	
\end{remark}

\section{Direct linearization scheme of the unreduced ASDYM}\label{sec-3}

In this section, we aim to describe a DL scheme for the unreduced ASDYM equation \eqref{ASDYM-unreduce}.
The idea is motivated by the link between the SDYM equation and the AKNS system that we
have found in \cite{LQYZ-SAPM-2022}.
We will establish the DL scheme for the negative and positive AKNS hierarchy
from which we are able to construct solutions of the unreduced ASDYM equation.
The DL scheme is composed of the following steps.

\subsection{Linear integral equations}\label{sec-3-1}

The DL scheme of our interest starts form the following integral equation set
(cf.\cite{Nijhoff-1983}):
\bsb\label{int-eq}
\begin{align}
	\label{int-eq-1}
\mf u(k)-\iint_D\mu(l,l')\frac{\rho_1(k)\sigma_2(l')}{k-l'}\cdot\frac{\rho_2(l')\sigma_1(l)}{l'-l}\mf u(l)
=\rho_1(k)\mf c_k, \\
	\label{int-eq-2}
\mf v(k)-\iint_D\mu(l,l')\frac{\rho_2(k)\sigma_1(l)}{k-l}\cdot\frac{\rho_1(l)\sigma_2(l')}{l-l'}\mf v(l')
=\rho_2(k)\mf c_k.
\end{align}
\esb
Here,
$\mf u(k)$ and $\mf v(k)$ are infinite-dimensional column vector functions of $\mf x=(\cdots,x_{-1},x_0,x_1,\cdots)$
and rely on a  parameter $k$,
while $\mf c_k=(\cdots,k^{-2},k^{-1},1,k,k^2,\cdots)^T$.
The plane wave factors are taken as
\bsb\label{PWF-3.2}
\begin{align}
	& \rho_1(l)=\exp\Big(\sum_{n\in\mb Z}l^n  x_n+\lambda_1(l)\Big),& &\rho_2(l')
=\exp\Big(-\sum_{n\in\mb Z}l'^n  x_n+\lambda_2(l')\Big), \\
	&\sigma_1(l)=\exp\Big(\sum_{n\in\mb Z}l^n  x_n+\gamma_1(l)\Big),& &\sigma_2(l')
=\exp\Big(-\sum_{n\in\mb Z}l'^n  x_n+\gamma_2(l')\Big),
\end{align}
\esb
where the initial values $\lambda_i(l)$ and $\gamma_i(l)$, $i=1,2$, are constants depending on $l$.
At this stage, the integral measure $\mu(l,l')$ can be general and the integration domain
$D$ is any suitably chosen subset in $\mathbb{C}^2$,
such that it holds as the general assumption
that for given measure and integration domain $D$ the solution of each integral equation is unique.
In practice, we assume that the measure and the integration domain
can be separated into two independent parts, i.e.
\begin{equation}\label{3.3}
	\mu(l,l')=\mu_1(l)\cdot\mu_2(l'),~~~~D=D_1\times D_2,
\end{equation}
where $D_1$ and $D_2$ are independent  contours   and $\l\in D_1$, $l'\in D_2$, respectively.
Such a setting allows us to decompose the double integral as
\begin{equation}\label{tech-1}
	\iint_D f(l)g(l') \mu(l,l')=\int_{D_1}f(l)\mu_1(l)\cdot\int_{D_2}g(l')\mu_2(l'),
\end{equation}
from which we proceed with the follow-up steps.

\subsection{Infinite-dimensional matrices}\label{sec-3-2}

To proceed, we introduce infinite-dimensional matrices $\mf I, \bs\Lambda$ and $\mf O$,
where their $(i,j)$-th entries are  respectively
\begin{align}
	(\mf I)_{i,j}\doteq\delta_{i,j},~~~~
	(\bs\Lambda)_{i,j}\doteq\delta_{i+1,j},~~~~(\mf O)_{i,j}\doteq\delta_{i,0}\delta_{0,j},
\end{align}
with
\begin{align*}
	\delta_{i,j}=
	\left\{
	\begin{array}{lll}
		1, && i=j,\\
		0,  && i\neq j.
	\end{array}\right.
\end{align*}
$\mf I$ is the identity matrix.
$\mf O$ is a projective matrix satisfying $\mf O^2=\mf O$,
and for any  column vectors
\begin{align*}
 \mf a\doteq(\cdots,a_{-1},a_{0},a_{1},\cdots)^T, ~~\mf b\doteq(\cdots,b_{-1},b_0,b_1,\cdots)^T.
\end{align*}
there is $\mf a^T\mf O\mf b=a_0b_0$.
Matrix $\bs\Lambda$ acts as a shift operator, e.g., for $\mf c_k$ defined in Sec.\ref{sec-3-1},
it holds that $\bs\Lambda\mf c_k =k \mf c_k$.
Next, using $\bs\Lambda$ and $\mf O$ we
define matrix $\bs\Omega$ by
\begin{align}\label{def-omega-1}
	\bs\Omega=\sum_{i=0}^{\infty}(\bs\Lambda^T)^{-i-1}\mf O\bs\Lambda^i.
\end{align}
It then follows that
\begin{align}
		\mf c_k^T\bs\Omega\mf c_l=\sum_{i=0}^{\infty}\mf c_k^T(\bs\Lambda^T)^{-i-1}\mf O\bs\Lambda^i\mf c_l
=	\sum_{i=0}^{\infty}\frac{l^i}{k^{i+1}}=\frac{1}{k}\Big(\frac{1}{1-\frac{l}{k}}\Big)=\frac{1}{k-l}.
\end{align}
This indicates  a skew symmetric relation $\bs\Omega=-\bs\Omega^T$,
which agrees with an equivalent definition of $\bs\Omega$, i.e.
\begin{align}\label{def-omega-2}
	\bs\Omega=-\sum_{i=0}^{\infty}(\bs\Lambda^T)^{i}\mf O\bs\Lambda^{-i-1}.
\end{align}
Note that $\bs\Omega$ satisfies
\begin{equation}\label{OOmega}
\bs\Lambda^T \bs\Omega -\bs\Omega\bs\Lambda =\mf O.
\end{equation}

In terms of these matrices, one can rewrite the linear integral equations in \eqref{int-eq}.
\begin{proposition}\label{Prop-0}
The linear integral equations \eqref{int-eq} can be rewritten as
\begin{subequations}\label{3.9}
\begin{align}\label{uk-1}
&	\mf u(k)=(\mf I+\mf U\bs\Omega\mf C_2\bs\Omega)\rho_1(k)\mf c_k,\\
&	\mf v(k)=(\mf I+\mf V\bs\Omega\mf C_1\bs\Omega)\rho_2(k)\mf c_k,  \label{vk-1}
\end{align}
\end{subequations}
where we  define
\begin{subequations}\label{3.10}
\begin{align}
&	\mf U\doteq\int_{D_1}\mu_1(l)\mf u(l)\mf c_{l}^T\sigma_1(l),~~~
	\mf V\doteq\int_{D_2}\mu_2(l')\mf v(l')\mf c^T_{l'}\sigma_2(l'),\label{def-U-C2}\\
&\mf C_1\doteq\int_{D_1}\mu_1(l)\rho_1(l)\mf c_{l}\mf c_{l}^T\sigma_1(l),~~~
  \mf C_2\doteq\int_{D_2}\mu_2(l')\rho_2(l')\mf c_{l'}\mf c_{l'}^T\sigma_2(l'). \label{def-V-C1}
\end{align}
\end{subequations}
\end{proposition}

\begin{proof}
	By replacing $\frac{1}{k-l}$ with $\mf c_k^T\bs\Omega\mf c_l$, one obtains
	\begin{align*}
		\mf u(k)-\iint_D\mu(l,l')\rho_1(k)\mf c_k^T\bs\Omega\mf c_{l'}\sigma_2(l')
 \rho_2(l')\mf c^T_{l'}\bs\Omega\mf c_l\sigma_1(l) \mf u(l)=\rho_1(k)\mf c_k.
	\end{align*}
It then  follows from \eqref{3.3} that
	\begin{align*}
		\mf u(k)-\int_{D_1}\mu(l)\mf u(l)\mf c_l^T\sigma_1(l)\,
\bs\Omega\int_{D_2}\mu(l')\rho_2(l')\mf c_{l'}\mf c_{l'}^T\sigma_2(l')
\bs\Omega\rho_1(k)\mf c_k=\rho_1(k)\mf c_k,
	\end{align*}
which gives rise to the form  \eqref{uk-1}.
\eqref{vk-1} can be obtained in a similar way.
\end{proof}

Note that $\mf U$ and $\mf V$ can be expressed in terms of $\bs\Omega$ and $\mf C_j$.
In fact, from \eqref{uk-1}, by integration we have
\begin{align}
	\int_{D_1}\zeta_1(k)\mf u(k)\mf c_{k}^T\sigma_1(k)
	=(\mf I+\mf U\bs\Omega\mf C_2\bs\Omega)\int_{D_1}\zeta_1(k)\rho_1(k)\mf c_k\mf c_{k}^T\sigma_1(k),
\end{align}
which, in light of   \eqref{3.10}, is
\begin{equation}
	\mf U=(\mf I+\mf U\bs\Omega\mf C_2\bs\Omega)\mf C_1,
\end{equation}
or alternatively
\begin{equation}\label{U-rep}
	\mf U=\mf C_1(\mf I-\bs\Omega\mf C_2\bs\Omega\mf C_1)^{-1}
=(\mf C_1^{-1}-\bs\Omega\mf C_2\bs\Omega)^{-1}.
\end{equation}
Similarly, we can obtain the following expressions for $\mf V$:
\begin{equation}\label{V-rep}
	\mf V=(\mf I+\mf V\bs\Omega\mf C_1\bs\Omega)\mf C_2
=(\mf C_2^{-1}-\bs\Omega\mf C_1\bs\Omega)^{-1}.
\end{equation}

\subsection{Some properties}\label{sec-3-3}

Later, we will  use $\mf U$ and $\mf V$ to construct solutions for the unreduced ASDYM equation.
In this part, let us  their properties.

\begin{proposition}\label{prop-1}
	Infinite-dimensional matrices $\mf U$ and $\mf V$ are symmetric matrices.
\end{proposition}
\begin{proof}
It is apparent from   \eqref{def-V-C1}
that both $\mf C_1$ and $\mf C_2$ are symmetric matrices.
Since $\bs\Omega^T=-\bs\Omega$, we then have
	\begin{align*}
		\mf U^T=((\mf C_1^{-1}-\bs\Omega\mf C_2\bs\Omega)^{-1})^T
		=(\mf C_1^{-1}-\bs\Omega\mf C_2\bs\Omega)^{-1}=\mf U, \\
		\mf V^T=((\mf C_2^{-1}-\bs\Omega\mf C_1\bs\Omega)^{-1})^T
		=(\mf C_2^{-1}-\bs\Omega\mf C_1\bs\Omega)^{-1}=\mf V.
	\end{align*}
	Therefore, matrices $\mf U$ and $\mf V$ are symmetric.
\end{proof}

\begin{proposition}\label{prop-2}
	Let $\partial_n$ be the abbreviation of $\frac{\partial}{\partial {x_n}}$.
	Then, calculating the derivatives of $\mf U$ and $\mf V$,
	we have the following evolution relations:
	\bsb\label{evo-UV}
	\begin{align}
		\label{evo-U}\partial_n\mf U&=\bs\Lambda^n\mf U+\mf U(\bs\Lambda^T)^n
+\overline{\mf U}\mf O_n\mf U-\mf U\mf O_n\overline{\mf V}, \\
		\label{evo-V}\partial_n\mf V&=-\bs\Lambda^n\mf V-\mf V(\bs\Lambda^T)^n
-\overline{\mf V}\mf O_n\mf V+\mf V\mf O_n\overline{\mf U},
	\end{align}
	\esb
	where
	\begin{equation}\label{3.16}
		\overline{\mf U}\doteq-\mf U\bs\Omega\mf C_2, ~~~\overline{\mf V}\doteq-\mf V\bs\Omega\mf C_1, ~~~
		\mf O_n\doteq(\bs\Lambda^T)^n\bs\Omega-\bs\Omega\bs\Lambda^n.
	\end{equation}
\end{proposition}

\begin{proof}
	We start by considering the evolution of $\mf C_1$.
	According to \eqref{def-V-C1} we have
	\begin{align*}
		\partial_n\mf C_1=\int_{D_1}\mu_1(l)\partial_{n}\rho_1(l)\mf c_{l}\mf c_{l}^T\sigma_1(l)
		+\int_{D_1}\mu_1(l)\rho_1(l)\mf c_{l}\mf c_{l}^T\partial_{n}\sigma_1(l)
		=2\int_{D_1}\mu_1(l)l^n\rho_1(l)\mf c_{l}\mf c_{l}^T\sigma_1(l),
	\end{align*}
	which can be represented as $\partial_n\mf C_1=2\bs\Lambda^n\mf C_1$.
From the definition in \eqref{def-V-C1},  $\mf C_1$ is a symmetric matrix of the following form
	\begin{align*}
		\mf C_1=\int_{D_1}\mu_1(l)\rho_1(l)\sigma_1(l)~
		\left[\begin{matrix}
			& \vdots & \vdots & \vdots & \vdots & \vdots &  \\
			\cdots & l^{-4} & l^{-3} & l^{-2} & l^{-1} & 1 & \cdots \\
			\cdots & l^{-3} & l^{-2} & l^{-1} & 1 & l & \cdots \\
			\cdots & l^{-2} & l^{-1} & 1 & l & l^{2} & \cdots \\
			\cdots & l^{-1} & 1 & l & l^{2} & l^{3} & \cdots \\
			\cdots & 1 & l & l^{2} & l^{3} & l^{4} & \cdots \\
			& \vdots & \vdots & \vdots & \vdots & \vdots &
		\end{matrix}\right],
	\end{align*}
which means $\partial_n\mf C_1$ is a symmetric matrix as well,
and so is $\bs\Lambda^n\mf C_1$.
Thus we have
\begin{equation}\label{LC}
\bs\Lambda^n\mf C_1=(\bs\Lambda^n\mf C_1)^T=\mf C_1(\bs\Lambda^T)^n,
\end{equation}
and consequently,
\begin{equation}\label{partial-C1}
		\partial_n\mf C_1=2\bs\Lambda^n\mf C_1=\bs\Lambda^n\mf C_1+\mf C_1(\bs\Lambda^T)^n.
\end{equation}
	There is a similar result for $\mf C_2$, which is
	\begin{equation}\label{partial-C2}
		\partial_{n}\mf C_2=-\bs\Lambda^n\mf C_2-\mf C_2(\bs\Lambda^T)^n.
	\end{equation}

Next, from \eqref{U-rep} we have
\begin{equation}\label{U-rep-2}
\mf U^{-1}= \mf C_1^{-1}-\bs\Omega\mf C_2\bs\Omega.
\end{equation}
Taking  $x_n$-derivative of the above $\mf U^{-1}$, one obtains
\begin{align*}
\partial_{n}\mf U^{-1}=-\mf U^{-1}(\partial_{n}\mf U)\mf U^{-1}
=-\mf C_1^{-1}(\partial_{n}\mf C_1 )\mf C_1^{-1}-\bs\Omega(\partial_n\mf C_2)\bs\Omega,
\end{align*}
which gives rise to
\begin{equation}\label{partial-U-half}
\partial_n\mf U=\mf U\mf C_1^{-1}(\partial_n\mf C_1)\mf C_1^{-1}\mf U
                        +\mf U\bs\Omega(\partial_n\mf C_2)\bs\Omega\mf U.
\end{equation}
Then, substituting \eqref{partial-C1} and \eqref{partial-C2} into \eqref{partial-U-half}
and employing the notation $\mf O_n$ as defined in \eqref{3.16}, we have
\begin{align*}
\partial_n\mf U
		&=\mf U\mf C_1^{-1}\bs\Lambda^n\mf U+\mf U(\bs\Lambda^T)^n\mf C_1\mf U
             -\mf U\bs\Omega\bs\Lambda^n\mf C_2\bs\Omega\mf U
             -\mf U\bs\Omega\mf C_2(\bs\Lambda^T)^n\bs\Omega\mf U \notag \\
		&=\mf U\mf C_1^{-1}\bs\Lambda^n\mf U+\mf U(\bs\Lambda^T)^n\mf C_1\mf U
             +\mf U(\mf O_n-(\bs\Lambda^T)^n\bs\Omega)\mf C_2\bs\Omega\mf U
             -\mf U\bs\Omega\mf C_2(\mf O_n+\bs\Omega\bs\Lambda^n)\mf U \notag \\
		&=\mf U(\mf C_1^{-1}-\bs\Omega\mf C_2\bs\Omega)\bs\Lambda^n\mf U
             +\mf U(\bs\Lambda^T)^n(\mf C_1-\bs\Omega\mf C_2\bs\Omega)\mf U
             +\mf U\mf O_n\mf C_2\bs\Omega\mf U-\mf U\bs\Omega\mf C_2\mf O_n\mf U \notag \\
		&=\bs\Lambda^n\mf U+\mf U(\bs\Lambda^T)^n+\mf U\mf O_n\mf C_2\bs\Omega\mf U
             -\mf U\bs\Omega\mf C_2\mf O_n\mf U\\
        &=\bs\Lambda^n\mf U+\mf U(\bs\Lambda^T)^n-\mf U\mf O_n\overline{\mf U}^T
             -\mf U\bs\Omega\overline{\mf U},
\end{align*}
where $\overline{\mf U}$ is defined in \eqref{3.16}
and we have made use of the symmetric relations $\mf U=\mf U^T, \mf C_2=\mf C_2^T$
and $\bs\Omega=-\bs\Omega^T$.
In addition, $\overline{\mf U}$ and $\overline{\mf V}$ are connected by
\begin{equation}\label{UV-bar}
\overline{\mf U}^T=-\overline{\mf V},
\end{equation}
which can be seen from
\begin{align*}
\overline{\mf U}^T=(-\mf U\bs\Omega\mf C_2)^T
=-(((\mf C_2\bs\Omega\mf C_1)^{-1}-\bs\Omega)^{-1})^T
=((\mf C_1\bs\Omega\mf C_2)^{-1}-\bs\Omega)^{-1}=\mf V\bs\Omega\mf C_1=-\overline{\mf V}.
\end{align*}
Thus, we are finally led to \eqref{evo-U}.
Another evolution \eqref{evo-V} can be derived through a similar procedure, and we skip it.

\end{proof}

A direct consequence of the above Proposition is the following.

\begin{proposition}\label{prop-2-coro-1}
The derivatives of $\overline{\mf U}$ and $\overline{\mf V}$ are given by
\bsb\label{evo-bar-UV}
\begin{align}
		\label{evo-bar-U}\partial_n\overline{\mf U}&=\bs\Lambda^n\overline{\mf U}
-\overline{\mf U}(\bs\Lambda^T)^n+\overline{\mf U}\mf O_n\overline{\mf U} -\mf U\mf O_n\mf V, \\
		\label{evo-bar-V}\partial_n\overline{\mf V}&=-\bs\Lambda^n\overline{\mf V}
+\overline{\mf V}(\bs\Lambda^T)^n-\overline{\mf V}\mf O_n\overline{\mf V}+\mf V\mf O_n\mf U.
	\end{align}
\esb
\end{proposition}

The infinite-dimensional matrices $\mf U, \mf V, \overline{\mf U}$ and $\overline{\mf U}$
are also connected without involving derivatives.

\begin{proposition}\label{prop-3}
$\mf U$ and $\mf V$ satisfy the following relations
\bsb\label{diff-uv-1}
\begin{align}
	\mf U(\bs\Lambda^T)^n-\bs\Lambda^n\mf U&=\overline{\mf U}\mf O_n\mf U+\mf U\mf O_n\overline{\mf V}, \\
	\mf V(\bs\Lambda^T)^n-\bs\Lambda^n\mf V&=\overline{\mf V}\mf O_n\mf V+\mf V\mf O_n\overline{\mf U}.
\end{align}
\esb
\end{proposition}

\begin{proof}
	We only prove the first case.
	The second one can be proved  similarly.
	Direct computations yield
\begin{align*}
	&~\bs\Lambda^n\mf U-\mf U(\bs\Lambda^T)^n+\overline{\mf U}\mf O_n\mf U+\mf U\mf O_n\overline{\mf V} \\
  =&~\bs\Lambda^n\mf U-\mf U(\bs\Lambda^T)^n
      -\overline{\mf U}(\bs\Omega\bs\Lambda^n
      -(\bs\Lambda^T)^n\bs\Omega)\mf U
		-\mf U(\bs\Omega\bs\Lambda^n
      -(\bs\Lambda^T)^n\bs\Omega)\overline{\mf V} \\
  =&~(\mf I-\overline{\mf U}\bs\Omega)\bs\Lambda^n\mf U
        -\mf U(\bs\Lambda^T)^n(\mf I-\bs\Omega\overline{\mf V})
		+\overline{\mf U}(\bs\Lambda^T)^n\bs\Omega\mf U
        -\mf U\bs\Omega\bs\Lambda^n\overline{\mf V} \\
 =&~\mf U\mf C_1^{-1}\bs\Lambda^n\mf U
        -\mf U(\bs\Lambda^T)^n\mf C_1^{-1}\mf U
		+\mf U\bs\Omega\bs\Lambda^n\mf C_2\bs\Omega\mf U
        -\mf U\bs\Omega\mf C_2(\bs\Lambda^T)^n\bs\Omega\mf U \\
 =&~\mf U\mf C_1^{-1}(\bs\Lambda^n\mf C_1-\mf C_1(\bs\Lambda^T)^n)\mf C_1^{-1}\mf U
       +\mf U\bs\Omega(\bs\Lambda^n\mf C_2-\mf C_2(\bs\Lambda^T)^n)\bs\Omega\mf U\\
 =&~\bs0.
\end{align*}
	This completes the proof.
\end{proof}

There are dual relations which can be proved in a similar way.
\begin{proposition}\label{prop-3-coro-2}
	$\overline{\mf U}$ and $\overline{\mf V}$ satisfy the following relations:
	\bsb\label{diff-uv-bar}
	\begin{align}
		\overline{\mf U}(\bs\Lambda^T)^n-\bs\Lambda^n\overline{\mf U}
&=\overline{\mf U}\mf O_n\overline{\mf U}+\mf U\mf O_n\mf V, \\
		\overline{\mf V}(\bs\Lambda^T)^n-\bs\Lambda^n\overline{\mf V}
&=\overline{\mf V}\mf O_n\overline{\mf V}+\mf V\mf O_n\mf U.
	\end{align}
	\esb
\end{proposition}

Here we insert a remark on $\mf O_n$.
\begin{remark}\label{rem-2}
The definition of $\mf O_n$ is an extension of \eqref{OOmega} if one regards $\mf O$ as $\mf O_1$,
and there is a recurrence relation for $\{\mf O_j\}$
\begin{equation}\label{rec-Oij}
	\mf O_{i+j}=\mf O_i\bs\Lambda^j+(\bs\Lambda^T)^i\mf O_j,~~~~(i,j\in\mb Z).
\end{equation}
Making use of this relation,
with appropriate computations,
one can obtain another expression of $\mf O_n$ represented by $\bs\Lambda$ and $\mf O$:
\begin{equation}\label{rep-on}
	\mf O_n = \left\{
	\begin{array}{lll}
		\sum_{l=0}^{n-1}(\bs\Lambda^T)^{n-1-l}\mf O\bs\Lambda^{l}, && n\in\mb Z^+,\\
		\bs0, && n=0,\\
		-\sum_{l=-1}^{n}(\bs\Lambda^T)^{n-1-l}\mf O\bs\Lambda^{l}, && n\in\mb Z^-.
	\end{array}\right.
\end{equation}
\end{remark}

Note that the results in Proposition \ref{prop-2}, \ref{prop-2-coro-1}, \ref{prop-3} and
\ref{prop-3-coro-2} can be alternatively expressed by introducing the following
formal $2\times 2$ block matrix
\begin{equation}\label{def-T}
	\mk T=
	\left[
	\begin{array}{cc}
		\overline{\mf V} &    \mf V    \\
		\mf U  &   \overline{\mf U}
	\end{array}
	\right],
\end{equation}
where $\mf U, \mf V, \overline{\mf U}, \overline{\mf V}$ are the infinite-dimensional matrices defined previously.

\begin{proposition}\label{prop-4}
The evolution of $\mk T$ is
	\begin{equation}\label{evo-T}
		\partial_n\mk T=\mk T (\mk A^{T})^n\mathcal{C}-\mk C\mk A^n \mk T-\mk T \mk C\mk O_n \mk T.
	\end{equation}
    In addition, $\mk T$ satisfies the `difference' relation:
	\begin{equation}\label{recur-1}
		\mk T(\mk A^{T})^n-\mk A^n\mk T=\mk T\mk O_n\mk T.
	\end{equation}
    Here $\mk C, \mk A$ and $\mk O_n$ are all formal block diagonal matrices:
	\begin{equation}\label{CAO}
		\mk C\doteq
		\left[
		\begin{array}{cc}
			\mf I &        \\
			&   -\mf I
		\end{array}
		\right],~~~~
		\mk A\doteq
		\left[
		\begin{array}{cc}
			\bs\Lambda &        \\
			&   \bs\Lambda
		\end{array}
		\right],~~~~
		\mk O_n\doteq
		\left[
		\begin{array}{cc}
			\mf O_n &        \\
			&   \mf O_n
		\end{array}
		\right].
	\end{equation}

\end{proposition}

\begin{proof}
Substituting \eqref{evo-UV} and \eqref{evo-bar-UV} into
	\begin{align}
		\partial_n\mk T=&
		\left[
		\begin{array}{cc}
			\partial_n\overline{\mf V} &    \partial_n\mf V    \\
			\partial_n\mf U  &   \partial_n\overline{\mf U}
		\end{array}
		\right],
	\end{align}
	we have
	\begin{align*}
		\partial_n\mk T=&
		\left[
		\begin{array}{cc}
			-\bs\Lambda^n\overline{\mf V}+\overline{\mf V}(\bs\Lambda^T)^n
            -\overline{\mf V}\mf O_n\overline{\mf V}+\mf V\mf O_n\mf U
            &    -\bs\Lambda^n\mf V-\mf V(\bs\Lambda^T)^n-\overline{\mf V}\mf O_n\mf V
                  +\mf V\mf O_n\overline{\mf U}    \\
			\bs\Lambda^n\mf U+\mf U(\bs\Lambda^T)^n+\overline{\mf U}\mf O_n\mf U-\mf U\mf O_n\overline{\mf V}
            &   \bs\Lambda^n\overline{\mf U}-\overline{\mf U}(\bs\Lambda^T)^n
                  +\overline{\mf U}\mf O_n\overline{\mf U} 	-\mf U\mf O_n\mf V
		\end{array}
		\right] \\
		=&\left[
		\begin{array}{cc}
			-\bs\Lambda^n\overline{\mf V}+\overline{\mf V}(\bs\Lambda^T)^n
          &    -\bs\Lambda^n\mf V-\mf V(\bs\Lambda^T)^n    \\
			\bs\Lambda^n\mf U+\mf U(\bs\Lambda^T)^n
          &   \bs\Lambda^n\overline{\mf U}-\overline{\mf U}(\bs\Lambda^T)^n
		\end{array}
		\right]
		-\left[
		\begin{array}{cc}
			\overline{\mf V} &    \mf V    \\
			\mf U  &   \overline{\mf U}
		\end{array}
		\right]\left[
		\begin{array}{cc}
			\mf O_n &        \\
			&   -\mf O_n
		\end{array}
		\right]\left[
		\begin{array}{cc}
			\overline{\mf V} &    \mf V    \\
			\mf U  &   \overline{\mf U}
		\end{array}
		\right] \\
		=&~\mk T
		\left[
		\begin{array}{cc}
			(\bs\Lambda^T)^n &      \\
			&   -(\bs\Lambda^T)^n
		\end{array}
		\right]-
		\left[
		\begin{array}{cc}
			\bs\Lambda^n &      \\
			&   -\bs\Lambda^n
		\end{array}
		\right] \mk T-\mk T
		\left[
		\begin{array}{cc}
			\mf O_n &        \\
			&   -\mf O_n
		\end{array}
		\right] \mk T \\
		=&~\mk T (\mk A^{T})^n\mk C-\mk C\mk A^n \mk T-\mk T \mk C\mk O_n \mk T,
	\end{align*}
which completes the proof for \eqref{evo-T}.

	For the relation \eqref{recur-1}, utilizing \eqref{diff-uv-1} and \eqref{diff-uv-bar},
	we can calculate $\mk T(\mk A^T)^n-\mk A^n\mk T$ to get
	\begin{align*}
		\mk T(\mk A^{T})^n-\mk A^n\mk T
      &=\left[
		\begin{array}{cc}
			\overline{\mf V}(\bs\Lambda^T)^n-\bs\Lambda^n\overline{\mf V}
      &   \mf V(\bs\Lambda^T)^n-\bs\Lambda^n\mf V     \\
			\mf U(\bs\Lambda^T)^n-\bs\Lambda^n\mf U
      &   \overline{\mf U}(\bs\Lambda^T)^n-\bs\Lambda^n\overline{\mf U}
		\end{array}
		\right] \\
		&=
		\left[
		\begin{array}{cc}
			\overline{\mf V} &    \mf V    \\
			\mf U  &   \overline{\mf U}
		\end{array}
		\right]\left[
		\begin{array}{cc}
			\mf O_n &        \\
			&   \mf O_n
		\end{array}
		\right]
		\left[
		\begin{array}{cc}
			\overline{\mf V} &    \mf V    \\
			\mf U  &   \overline{\mf U}
		\end{array}
		\right]
		=\mk T\mk O_n\mk T.
	\end{align*}
	This completes the proof for \eqref{recur-1}.

\end{proof}

The relations \eqref{evo-T} and \eqref{recur-1} allow us to build up the relation
between $\partial_{n+s}\mk T$ and $\partial_{n}\mk T$,
which provides a differential recurrence for $\mk T$.

\begin{theorem}\label{Th-1}
The following differential recurrence relation of $\mk T$ holds,
\begin{equation}\label{diff-re-T}
(\partial_{n+s}\mk T)(\mk A^{\mr T})^{-s}+(\partial_n\mk T) \mk O_s \mk T(\mk A^{\mr T})^{-s}
=\partial_n\mk T,~~~
n,s\in\mb Z, ~~s\neq0.
\end{equation}
\end{theorem}

\begin{proof}
First, noticing that from the relation \eqref{rec-Oij} and the notations in \eqref{CAO}, we have
\begin{equation}\label{rec-Ons}
	\mk O_{n+s}=\mk O_n\mk A^s+(\mk A^T)^n\mk O_s, ~~~ n,s\in\mb Z
\end{equation}
and $\mk C\mk A^T=\mk A^T \mk C$.
Then, using \eqref{evo-T} and the above relations, we have
\begin{align*}
(\partial_{n+s}\mk T)(\mk A^{T})^{-s}
	=&\mk T(\mk A^{T})^{n}\mk C-\mk C\mk A^{n+s}\mk T(\mk A^{T})^{-s}
           -\mk T \mk C\mk O_{n+s} \mk T(\mk A^{T})^{-s} \\
	=&\mk T(\mk A^{T})^{n}\mk C-\mk C\mk A^{n+s}\mk T(\mk A^{T})^{-s}
           -\mk T \mk C\mk O_{n} \mk A^s\mk T(\mk A^{T})^{-s}
           -\mk T(\mk A^{T})^n \mk C\mk O_{s} \mk T(\mk A^{T})^{-s}.
\end{align*}
As for the term $(\partial_n\mk T) \mk O_s \mk T(\mk A^{\mr T})^{-s}$,
we can substitute \eqref{evo-T} into it to obtain:
\begin{align*}
	&~(\partial_n\mk T) \mk O_s \mk T(\mk A^{T})^{-s}  \\
	=&~\mk T(\mk A^{T})^n \mk C\mk O_s \mk T(\mk A^{T})^{-s}
        -\mk C\mk A^n(\mk T\mk O_s\mk T)(\mk A^{T})^{-s}
		-\mk T \mk C\mk O_n (\mk T\mk O_s\mk T)(\mk A^{T})^{-s} \\
	=&~\mk T (\mk A^{T})^n \mk C\mk O_s \mk T(\mk A^{T})^{-s}
        -\mk C\mk A^n\mk T+\mk C\mk A^{n+s}\mk T(\mk A^{T})^{-s}
		-\mk T \mk C\mk O_n \mk T+\mk T \mk C\mk O_n \mk A^s\mk T(\mk A^{T})^{-s},
\end{align*}
where in the last step we have used \eqref{recur-1} to replace the term  $\mk T\mk O_s\mk T$.
Finally, adding them together leads to
\begin{align*}
(\partial_{n+s}\mk T)(\mk A^{T})^{-s}+\partial_n\mk T \mk O_s \mk T(\mk A^{T})^{-s}
=\mk T(\mk A^{T})^{n}\mk C-\mk C\mk A^n\mk T-\mk T \mk C\mk O_n \mk T=\partial_n\mk T,
\end{align*}
where, again, in the final step we used \eqref{evo-T}.

\end{proof}

\subsection{Unreduced ASDYM equation}\label{sec-3-4}

We need to extract some $2\times 2$ matrix relations  from \eqref{diff-re-T}
so that we can obtain a unreduced ASDYM equation, from which one can go further to arrive at the ASDYM equation.
To achieve that, let us introduce a $2\times2$ matrix:
\begin{equation}\label{T-ij}
	\mk T^{(i,j)}=
	 \begin{pmatrix}
		(\overline{\mf V})_{i,j} & (\mf V)_{i,j} \\
		(\mf U)_{i,j} & (\overline{\mf U})_{i,j}
	\end{pmatrix},
\end{equation}
where by $(\mf A)_{i,j}$ we denote a scalar that stands for the $(i,j)$-th component of
the infinite-dimensional matrix $\mf A$.
Thus, from \eqref{diff-re-T} we can extract relations for $\{\mk T^{(i,j)}\}$.

\begin{theorem}\label{Th-2}
For the differential recurrence relation \eqref{diff-re-T} and the $2 \times 2$ matrices $\{\mk T^{(i,j)}\}$,
the following relations hold:
\bsb\label{diff-expand-T}
\begin{align}
\label{diff-expand-T-1}
&\partial_{n+s} \mk T^{(i,j-s)}+\sum_{l=0}^{s-1}(\partial_n\mk T^{(i,s-1-l)}) \mk T^{(l,j-s)}
=\partial_n \mk T^{(i,j)}, ~~~(s\in\mb Z^+), \\
\label{diff-expand-T-2}
&\partial_{n+s} \mk T^{(i,j-s)} -\sum_{l=-1}^{s}(\partial_n\mk T^{(i,s-1-l)}) \mk T^{(l,j-s)}
=\partial_n \mk T^{(i,j)}, ~~~(s\in\mb Z^-).
\end{align}
\esb
\end{theorem}

\begin{proof}
For the term $(\partial_{n+s}\mk T)(\mk A^{T})^{-s}$ in \eqref{diff-re-T},
since
\begin{align}
	(\partial_{n+s}\mk T)(\mk A^{T})^{-s}=
	\left[
	\begin{array}{cc}
		\partial_{n+s}\overline{\mf V}(\bs\Lambda^T)^{-s} &    \partial_{n+s}\mf V(\bs\Lambda^T)^{-s}    \\
		\partial_{n+s}\mf U(\bs\Lambda^T)^{-s}  &   \partial_{n+s}\overline{\mf U}(\bs\Lambda^T)^{-s}
	\end{array}
	\right],
\end{align}
we have
\begin{align}\label{partial-n+s-A-s}
	\partial_{n+s}(\mk T(\mk A^{T})^{-s})^{(i,j)}=
	\left[\begin{matrix}
		(\partial_{n+s}\overline{\mf V})_{i,j-s} & (\partial_{n+s}\mf V)_{i,j-s} \\
		(\partial_{n+s}\mf U)_{i,j-s} & (\partial_{n+s}\overline{\mf U})_{i,j-s}
	\end{matrix}\right]
	=\partial_{n+s} \mk T^{(i,j-s)}.
\end{align}

To deal with the term  of $(\partial_n\mk T) \mk O_s \mk T(\mk A^{T})^{-s}$ in \eqref{diff-re-T},
we first recall an identity  (see \cite{FW-thesis}):
\begin{equation}\label{lem-1-eq}
(\bs\Lambda^{m_1}\mf A(\bs\Lambda^T)^{n_1} \mf O \bs\Lambda^{m_2}\mf B(\bs\Lambda^T)^{n_2})_{i,j}
=(\mf A)_{i+m_1,n_1}(\mf B)_{m_2,j+n_2},~~~
i,j,m_1,m_2,n_1,n_2\in\mb Z,
\end{equation}
where $\mf A$ and $\mf B$ be two arbitrary infinite-dimensional matrices.
Applying \eqref{rep-on} we further have
\begin{equation}\label{AonB}
	(\mf A\mf O_n\mf B)_{i,j} = \left\{
	\begin{array}{lll}
		\sum_{l=0}^{n-1}(\mf A)_{i,n-1-l}(\mf B)_{l,j}, && n\in\mb Z^+,\\
		0, && n=0,\\
		-\sum_{l=-1}^{n}(\mf A)_{i,n-1-l}(\mf B)_{l,j}, && n\in\mb Z^-,
	\end{array}\right.
\end{equation}
from which we immediately find
\begin{equation}\label{TOsTA-s}
	((\partial_n\mk T) \mk O_s \mk T(\mk A^{T})^{-s})^{(i,j)}
= \left\{
	\begin{array}{lll}
		\sum_{l=0}^{s-1}(\partial_n\mk T^{(i,s-1-l)}) \mk T^{(l,j-s)}, && s\in\mb Z^+,\\
		-\sum_{l=-1}^{s}(\partial_n\mk T^{(i,s-1-l)}) \mk T^{(l,j-s)}, && s\in\mb Z^-.
	\end{array}\right.
\end{equation}
Finally, we arrive at \eqref{diff-expand-T} from \eqref{diff-re-T}, \eqref{partial-n+s-A-s} and \eqref{TOsTA-s}.

\end{proof}

To have the unreduced ASDYM equation, we introduce two variables
\begin{equation}\label{3.43}
	\bs P\doteq\bs I_2- \mk T^{(0,-1)}, ~~~\bs Q\doteq \mk T^{(0,0)},
\end{equation}
which are $2 \times 2$ matrices.
Take $i=j=0$, $s=1$ in \eqref{diff-expand-T-1},
 we have the following differential recurrence relation:
\begin{equation}\label{diff-recur}
	(\partial_{n+1}\bs P)\bs P^{-1}=-\partial_n\bs Q.
\end{equation}
Note that such a formula was also obtained in \cite{LQYZ-SAPM-2022} from the Cauchy matrix approach,
which serves as a crucial step for deriving the SDYM equation.
Through the compatibility $\partial_m(\partial_n\bs Q)=\partial_n(\partial_m\bs Q)$, one obtains
\begin{equation}\label{Yang-J-SDYM-1}
	\partial_m((\partial_{n+1}\bs P)\bs P^{-1})-\partial_n((\partial_{m+1}\bs P)\bs P^{-1})=0,
\end{equation}
or alternatively,
\begin{equation}\label{Yang-J-SDYM-2}
	\partial_{n+1}(\bs P^{-1}\partial_{m}\bs P)-\partial_{m+1}(\bs P^{-1}\partial_{n}\bs P)=0.
\end{equation}
Note that the equation \eqref{Yang-J-SDYM-2} is exactly in the form of \eqref{ASDYM-unreduce}.
We call \eqref{Yang-J-SDYM-2} the unreduced ASDYM equation.

We end this section with a property of $\bs P$.
We can show that $\bs P\in\mr{SL}(2)$. The determinant of $\bs P$ can be calculated directly as the following,
\begin{align}
		|\bs P|&=|\bs I_2- \mk T^{(0,-1)}|
                   =1-\mr{tr}( \mk T^{(0,-1)})+| \mk T^{(0,-1)}| \nonumber \\
		&=1-(-(\overline{\mf U})_{-1,0}+(\overline{\mf U})_{0,-1})
		+(-(\overline{\mf U})_{-1,0}(\overline{\mf U})_{0,-1}-(\mf V)_{-1,0}(\mf U)_{0,-1}) \nonumber \\
		&=1+(\bs\Lambda^{-1}\overline{\mf U}-\overline{\mf U}(\bs\Lambda^T)^{-1}
		-\bs\Lambda^{-1}\overline{\mf U}\mf O\overline{\mf U}(\bs\Lambda^T)^{-1}
		-\bs\Lambda^{-1}\mf V\mf O\mf U(\bs\Lambda^T)^{-1})_{0,0} \nonumber \\
		&=1+(\overline{\mf U}\bs\Lambda^T-\bs\Lambda\overline{\mf U}
		-\overline{\mf U}\mf O\overline{\mf U}
		-\mf V\mf O\mf U)_{-1,-1} =1. \label{P=1}
\end{align}

\section{Reduction to ASDYM}\label{sec-4}

In this section we reduce the unreduced ASDYM equation \eqref{Yang-J-SDYM-2} to the SU(2) ASDYM equation
in the  Euclidean space $\mathbb{E}$ and the ultrahyperbolic spaces $\mathbb{U}_1$ and $\mathbb{U}_2$.
The reduction is actually a procedure to fulfill the reality conditions in Table \ref{Tab-1}
and the gauge conditions in Table \ref{Tab-2}.

\subsection{Reduction to  the SU(2) ASDYM equation in $\mb E$}\label{sec-4-1}

\subsubsection{Reality condition}\label{sec-4-1-1}

For the SU(2) ASDYM equation in the Euclidean space $\mb E$,
to meet the reality condition in Table \ref{Tab-1},
we take $m=-n-1$ in \eqref{Yang-J-SDYM-2},
and take the following coordinate transformation:
\begin{equation}\label{coor-1}
	z_n\doteq x_n=\xi_n+\mr i\eta_n, ~~~\tilde z_n\doteq (-1)^{n+1}x_{-n}=\xi_n-\mr i\eta_n, ~~~n=1,2,\cdots,
\end{equation}
where $\xi_n$ and $\eta_n$ are real coordinates.
In this setting, equation \eqref{Yang-J-SDYM-2} reduces to
\begin{equation}
\partial_{\tilde z_{n+1}}(\bs P^{-1}\partial_{z_{n+1}}\bs P)+\partial_{z_n}(\bs P^{-1}\partial_{\tilde z_n}\bs P)=0.
\end{equation}

\subsubsection{Gauge condition}\label{sec-4-1-2}

Let us first fix the measures $\mu_i(l)$ in \eqref{3.10}.
Assume $\mathbb D_i$ is the domain with the contour $D_i$ as boundary
and $\{k_j\}_{j=1}^{M}\in \mathbb D_1,~ \{l_j\}_{j=1}^{M}\in \mathbb D_2$.
The  measures $\mu_1(l)$ and $\mu_2(l')$ are taken as the following,
\begin{equation}\label{measure}
\mu_1(l)=\sum^M_{j=1}\delta(l-k_j),~~
\mu_2(l')=\sum^M_{j=1}\delta(l'-l_j),
\end{equation}
where $\delta(l)$ is the formal Dirac $\delta$-function.
With the above $\{k_j, l_j\}$,  the plane wave factors in \eqref{int-eq} for the involved
variables $\{z_n,\tilde z_n\}$)
can be presented in terms of $\{\xi_m, \eta_m\}$ as follows,
\begin{align}
	&\rho_1(k_j)=\exp\big(\zeta_n(k_j)+\lambda_1(k_j) \big),    ~~ &
      \rho_2(l_j)=\exp\big(-\zeta_n(l_j)+\lambda_2(l_j) \big),\\
	&\sigma_1(k_j)=\exp\big(\zeta_n(k_j)+\gamma_1(k_j) \big), ~~ &
      \sigma_2(l_j)=\exp\big(-\zeta_n(l_j)+ \gamma_2(l_j)\big),
\end{align}
where
\begin{equation}
\zeta_n(k_j)=\sum_{m=n}^{n+1}\big((k_j^m+(-1)^{m+1}k_j^{-m})\xi_m
         +\mr i(k_j^m-(-1)^{m+1}k_j^{-m})\eta_m \big).
\end{equation}
With the above setting for measures, the matrices $\mf C_1$ and $\mf C_2$
in \eqref{def-V-C1} take the following explicit form
\begin{equation}\label{C1C2}
\mf C_1= \sum_{j=1}^{M} \rho_1(k_j)\sigma_1(k_j) \mf c_{k_j}\mf c_{k_j}^T, ~~~
\mf C_2= \sum_{j=1}^{M} \rho_2(l_j)\sigma_1(l_j) \mf c_{l_j}\mf c_{l_j}^T.
\end{equation}

To meet the gauge condition in Table \ref{Tab-2},
noting that we have shown $|\bs P|=1$ in \eqref{P=1} in the previous section,
in the following we focusing on achieving $\bs P=\bs P^\dag$.

We introduce the following constraints
\begin{subequations}\label{reduction-1-1-exp}
\begin{equation}\label{reduction-1}
	l_j=-\frac{1}{k_j^*}, ~~~j=1,2,\cdots,M,
\end{equation}
and we choose $\{\lambda_i(l_j), \lambda_i(k_j),\gamma_i(l_j), \gamma_i(k_j)\}$
such that
\begin{equation}\label{reduction-1-exp}
\exp\Big(\lambda_2(l_j)\Big)=\Big(\frac{\exp(\lambda_1(k_j))}{k_j}\Big)^*,~~~
 \exp\Big(\gamma_2(l_j)\Big)=-\Big(\frac{\exp(\gamma_1(k_j))}{k_j}\Big)^*, ~~~j=1,2,\cdots,M,
\end{equation}
\end{subequations}
which give rise to the relations
\begin{equation}\label{reduction-rho-sig}
\rho_2(l_j)=\Big(\frac{\rho_1(k_j)}{k_j}\Big)^*, ~~~
\sigma_2(l_j)=-\Big(\frac{\sigma_1(k_j)}{k_j}\Big)^*,
	 ~~~~j=1,2,\cdots,M.
\end{equation}
In fact, the constraint \eqref{reduction-1} yields
\begin{equation}
(\zeta_n(k_j))^*=-\zeta_n(l_j),
\end{equation}
which, together with \eqref{reduction-1-exp},
gives rise to \eqref{reduction-rho-sig}.

In the following, for the involved arguments, we start to examine their conjugate relations resulted from
the constraint \eqref{reduction-1-1-exp}.
First, let us look at $\mf c_{k_j}$ and $\mf c_{l_j}$.
It apparently follows from \eqref{reduction-1} that
\begin{align}
	(\mf c_{k_j}^T)^*&=(\cdots,k_j^{-2},k_j^{-1},1,k_j,k_j^2,\cdots)^* \notag \\
&=(\cdots,l_j^2,-l_j,1,-l_j^{-1},l_j^{-2},\cdots) \notag \\
&=(\cdots,l_j^2,l_j,1,l_j^{-1},l_j^{-2},\cdots) \,\bs\Xi, \notag \\
&=\mf c_{l_j}^T\bs\Theta\,\bs\Xi, 	\label{ck-trans}
\end{align}
where
\begin{equation*}
	\bs\Xi\doteq\mr{diag}(\cdots,1,-1,1,-1,1,\cdots), ~~~\text{i.e.} ~~~(\bs\Xi)_{i,j}\doteq (-1)^{i+j}\delta_{i,j},
\end{equation*}
and $\bs\Theta$ is a skew form of the identity matrix $\bs I$, i.e.
\begin{equation}
\bs \Theta =((\bs\Theta)_{i,j})_{\infty\times \infty},
~~~ (\bs\Theta)_{i,j}=\delta_{-i,j}.
\label{Theta}
\end{equation}

Before we proceed, let us collect some properties related to $\bs\Xi$ and $\bs\Theta$.
For example, it is easy to see that
\begin{align}
	\bs\Xi^2=\mf I, ~~~
	\bs\Theta^2=\mf I, ~~~\bs\Theta=\bs\Theta^T, ~~~
\bs\Theta\bs \Lambda=\bs \Lambda^T\bs\Theta, ~~~\bs\Lambda\bs\Xi=-\bs\Xi\bs\Lambda.
\end{align}
In addition, 
according to the expansion of $\bs\Omega$ (see \eqref{def-omega-1}\eqref{def-omega-2}) and the above formulae,
we have
\begin{subequations}\label{4.12}
\begin{align}
& \bs\Omega\bs\Xi= -\bs\Xi\bs\Omega, \label{4.12a}\\
& \bs\Theta\bs\Omega\bs\Theta=-\bs\Lambda^T\bs\Omega\bs\Lambda. \label{4.12b}
\end{align}
\end{subequations}
In fact, using \eqref{def-omega-1} and by direct calculation we have
\begin{align*}
	\bs\Omega\bs\Xi&=\sum_{i=0}^{\infty}(\bs\Lambda^T)^{-i-1}\mf O\bs\Lambda^i\bs\Xi
	=\sum_{i=0}^{\infty}(-1)^i(\bs\Lambda^T)^{-i-1}\mf O\bs\Xi\bs\Lambda^i \notag \\
	&=-\bs\Xi\sum_{i=0}^{\infty}(\bs\Lambda^T)^{-i-1}\mf O\bs\Lambda^i=-\bs\Xi\bs\Omega,
\end{align*}
and
\begin{align*}
	\bs\Theta\bs\Omega\bs\Theta &=\sum_{i=0}^{\infty}\bs\Theta(\bs\Lambda^T)^{-i-1}
     \mf O\bs\Lambda^i\bs\Theta \notag \\
	&=\sum_{i=0}^{\infty}(\bs\Lambda^T)^{i+1}\mf O\bs\Lambda^{-i}
	=\bs\Lambda^T\Big(\sum_{i=0}^{\infty}(\bs\Lambda^T)^{i}\mf O\bs\Lambda^{-i-1}\Big)\bs\Lambda
	=-\bs\Lambda^T\bs\Omega\bs\Lambda.
\end{align*}

Based on the above preparations, from the expression \eqref{C1C2}, the conjugate of $\mf C_1$ is then given by
\begin{align*}
\mf C_1 ^*
&=\Big(\sum_{j=1}^{N}\frac{\rho_1(k_j)}{k_j}k_j\mf c_{k_j}\mf c_{k_j}^Tk_j\frac{\sigma_1(k_j)}{k_j}\Big)^* \\
&=\sum_{j=1}^{N}\Big(\frac{\rho_1(k_j)}{k_j}\Big)^* \bs\Lambda \mf c_{k_j}^*
      (\mf c_{k_j}^T)^* \bs\Lambda^T \Big(\frac{\sigma_1(k_j)}{k_j}\Big)^*.
\end{align*}
Then, using \eqref{ck-trans} we have
\begin{align}
\mf C_1 ^*
	&=-\bs\Lambda\sum_{j=1}^N \rho_2(l_j)\bs\Xi\bs\Theta\mf c_{l_j}
        \mf c_{l_j}^T\bs\Theta\bs\Xi\sigma_2(l_j)\bs\Lambda^T
	=-\bs\Xi\bs\Lambda\bs\Theta\mf C_{2}\bs\Theta\bs\Lambda^T\bs\Xi,
\end{align}
where the anti-commutative relation $\bs\Lambda\bs\Xi=-\bs\Xi\bs\Lambda$ has been utilized.
Similarly, we have
\begin{equation}
	\mf C_2^*=-\bs\Xi\bs\Lambda\bs\Theta\mf C_{1}\bs\Lambda^T\bs\Theta\bs\Xi.
\end{equation}

Next, we come to look at the conjugate relation of $\mf U$ and $\mf V$.
Recalling the formula of $\mf U$ in \eqref{U-rep},
a direct substitution yields
\begin{align*}
	\mf U^*&=((\mf C_1^*)^{-1}-\bs\Omega\mf C_2^*\bs\Omega)^{-1} \\
	&=(-(\bs\Xi\bs\Lambda\bs\Theta\mf C_2\bs\Theta\bs\Lambda^T\bs\Xi)^{-1}+
	\bs\Omega\bs\Xi\bs\Lambda\bs\Theta\mf C_1\bs\Theta\bs\Lambda^T\bs\Xi\bs\Omega)^{-1}.
\end{align*}
Noticing that $\bs\Omega\bs\Xi=- \bs\Xi \bs\Omega$ and $\bs\Xi^{-1}=\bs \Xi$, we find
\begin{align*}
\mf U^*&= -\bs\Xi(\bs\Lambda\bs\Theta\mf C_2 \bs\Theta\bs\Lambda^T
                     -\bs\Omega\bs\Lambda\bs\Theta\mf C_1\bs\Theta\bs\Lambda^T\bs\Omega )^{-1}\bs\Xi \\
	&= -\bs\Xi(\bs\Lambda\bs\Theta\mf C_2 \bs\Theta\bs\Lambda^T
         -\bs\Lambda \bs\Lambda^T\bs\Omega\bs\Lambda\bs\Theta
           \mf C_1\bs\Theta\bs\Lambda^T\bs\Omega\bs\Lambda \bs\Lambda^T)^{-1}\bs\Xi,
\end{align*}
where we also utilized the relation $\bs\Lambda \bs\Lambda^T=\bs I$
which is true in the infinitely dimensional case.
Then, using \eqref{4.12b}, $\bs \Theta^2=\bs I$ and $\bs\Theta\bs \Lambda=\bs \Lambda^T\bs\Theta$
successively, we have
\begin{align*}
	\mf U^*&= -\bs\Xi(\bs\Lambda\bs\Theta\mf C_2 \bs\Theta\bs\Lambda^T
	-\bs\Lambda \bs \Theta \bs\Omega\mf C_1\bs\Omega\bs\Theta\bs\Lambda^T)^{-1}\bs\Xi \\
	&=-\bs\Xi(\bs\Theta\bs\Lambda^T)^{-1}(\mf C_2^{-1}
	- \bs\Omega\mf C_1\bs\Omega)^{-1}(\bs\Lambda\bs\Theta)^{-1}\bs\Xi  \\
	&=-\bs\Xi\bs\Lambda\bs\Theta\mf V\bs\Theta\bs\Lambda^{T}\bs\Xi,
\end{align*}
which indicates
\begin{align}
	(\mf U^*)_{i,j}=-(\bs\Xi\bs\Lambda\bs\Theta\mf V\bs\Theta\bs\Lambda^T\bs\Xi)_{i,j}
	=(-1)^{i+j+1}(\bs\Theta\mf V\bs\Theta)_{i+1,j+1}=(-1)^{i+j+1}(\mf V)_{-i-1,-j-1}.
\end{align}
As a result,
\begin{equation}
	(\mf U^*)_{0,-1}=(\mf V)_{-1,0}=(\mf V)_{0,-1},
\end{equation}
as $\mf V$ is symmetric.
The conjugate of $\overline{\mf U}$ can be calculated as well, which is
\begin{align*}
	(\overline{\mf U})^*&=-((\mf C_2^*\bs\Omega\mf C_1^*)^{-1}-\bs\Omega)^{-1} \\
	&=-((\bs\Xi\bs\Lambda\bs\Theta\mf C_1\bs\Theta\bs\Lambda^T\bs\Xi \bs\Omega
	\bs\Xi\bs\Lambda\bs\Theta\mf C_2\bs\Theta\bs\Lambda^T\bs\Xi)^{-1}-\bs\Omega)^{-1} \\
	&=((\bs\Xi\bs\Lambda\bs\Theta\mf C_1\bs\Theta\bs\Lambda^T\bs\Omega
	\bs\Lambda\bs\Theta\mf C_2\bs\Theta\bs\Lambda^T\bs\Xi)^{-1}-\bs\Xi\bs\Lambda\bs\Lambda^T
       \bs\Omega\bs\Lambda\bs\Lambda^T\bs\Xi)^{-1} \\
	&=-(\bs\Xi\bs\Lambda(\bs\Theta\mf C_1\bs\Omega
       \mf C_2\bs\Theta)^{-1}\bs\Lambda^T\bs\Xi-\bs\Xi\bs\Lambda\bs\Theta
       \bs\Omega\bs\Theta\bs\Lambda^T\bs\Xi)^{-1} \\
	&=-\bs\Xi\bs\Lambda\bs\Theta(({\mf C_1}{\bs\Omega}
        {\mf C_2})^{-1}-{\bs\Omega})^{-1}\bs\Theta\bs\Lambda^T\bs\Xi \\
	&=\bs\Xi\bs\Lambda\bs\Theta\overline{\mf V}\bs\Theta\bs\Lambda^T\bs\Xi.
\end{align*}
Hence we have
\begin{equation}
	(\overline{\mf U}^*)_{i,j}=(-1)^{i+j}(\overline{\mf V})_{-i-1,-j-1},
\end{equation}
as well as
\begin{equation}
	(\overline{\mf U}^*)_{0,-1}=-(\overline{\mf V})_{-1,0}=(\overline{\mf U})_{0,-1},~~~~
	(\overline{\mf V}^*)_{0,-1}=-(\overline{\mf U})_{-1,0}=(\overline{\mf V})_{0,-1}.
\end{equation}

Recalling the $2\times 2$ matrix $\bs P=\bs I_2-\mk T^{(0,-1)}$ that we have defined in \eqref{3.43},
utilizing the conjugate relations we have found above, we arrive at
\begin{align*}
	\bs P^*=
	\begin{pmatrix}
		1-(\overline{\mf V}^*)_{0,-1} & -(\mf V^*)_{0,-1} \\
		-(\mf U^*)_{0,-1} & 1-(\overline{\mf U}^*)_{0,-1}
	\end{pmatrix}
	=
	\begin{pmatrix}
		1-(\overline{\mf V})_{0,-1} & -(\mf U)_{0,-1} \\
		-(\mf V)_{0,-1} & 1-(\overline{\mf U})_{0,-1}
	\end{pmatrix}
	=\bs P^T,
\end{align*}
which concludes that $\bs P$ is a Hermitian matrix.

\subsection{Reduction to the SU(2) ASDYM equation in $\mb U_1$}\label{sec-4-2}

The reduction should give rise to the results that meet both the reality condition and the gauge condition.
In the unltrahyperbolic space $\mb U_1$, to agree with the reality condition (see Table \ref{Tab-1}),
we take  $m=-n-1$ in the unreduced ASDYM equation \eqref{Yang-J-SDYM-2}
and take the following coordinate transformation:
\begin{equation}\label{trans-U1}
	z_n\doteq x_n=\xi_n+\mr i\eta_n, ~~~
	\tilde z_n\doteq -x_{-n}=\xi_n-\mr i\eta_n, ~~~ 	n=1,2,\cdots,
\end{equation}
where $\xi_n$ and $\eta_n$ are real. The coordinate transformation leads \eqref{Yang-J-SDYM-2} to
\begin{equation}
\partial_{\tilde z_{n+1}}(\bs P^{-1}\partial_{z_{n+1}}\bs P)-\partial_{z_n}(\bs P^{-1}\partial_{\tilde z_n}\bs P)=0.
\end{equation}

Under the setting \eqref{measure} for the measures $\mu_1(l)$ and $\mu_2(l')$,
corresponding to the plane wave factors in \eqref{PWF-3.2} and the coordinate transformation \eqref{trans-U1},
we have the following plane wave factors:
\begin{subequations}\label{PWF-U1}
\begin{align}
	&\rho_1(k_j)=\exp\big(\bar\zeta_n(k_j)+\lambda_1(k_j) \big),    ~~ &
      \rho_2(l_j)=\exp\big(-\bar\zeta_n(l_j)+\lambda_2(l_j) \big),\\
	&\sigma_1(k_j)=\exp\big(\bar\zeta_n(k_j)+\gamma_1(k_j) \big), ~~ &
      \sigma_2(l_j)=\exp\big(-\bar\zeta_n(l_j)+ \gamma_2(l_j)\big),
\end{align}
where
\begin{equation}
\bar\zeta_n(k_j)=\sum_{m=n}^{n+1}\big((k_j^m-k_j^{-m})\xi_m+\mr i(k_j^m+k_j^{-m})\eta_m \big);
\end{equation}
\end{subequations}
the matrices $\mf C_1$ and $\mf C_2$ take the   form \eqref{C1C2}
but with the above plane wave factors.

To meet the gauge condition  $\bs P=\bs P^\dag$, we introduce constraints
\begin{subequations}\label{reduction-1-2-exp}
\begin{equation}\label{reduction-2}
	l_j=\frac{1}{k_j^*}, ~~~j=1,2,\cdots,M,
\end{equation}
and choose $\{\lambda_i(l_j), \lambda_i(k_j),\gamma_i(l_j), \gamma_i(k_j)\}$,  such that
\begin{equation}\label{reduction-2-exp}
\exp\Big(\lambda_2(l_j)\Big)=\Big(\frac{\exp(\lambda_1(k_j))}{k_j}\Big)^*,~~~
 \exp\Big(\gamma_2(l_j)\Big)= \Big(\frac{\exp(\gamma_1(k_j))}{k_j}\Big)^*, ~~~j=1,2,\cdots,M,
\end{equation}
\end{subequations}
which give rise to the relation
$(\bar\zeta_n(k_j))^*=-\bar\zeta_n(l_j)$
and then
\begin{equation}\label{reduction-rho-sig-2}
\rho_2(l_j)=\Big(\frac{\rho_1(k_j)}{k_j}\Big)^*, ~~~
\sigma_2(l_j)=\Big(\frac{\sigma_1(k_j)}{k_j}\Big)^*,
	 ~~~~j=1,2,\cdots,M.
\end{equation}

Similar to the Euclidean space case in the previous subsection,
we can successfully examine  conjugate relations of the involved arguments.
For $\mf c_{k_j}$ and $\mf c_{l_j}$, we have
\begin{align*}
(\mf c_{k_j}^T)^* =(\cdots,k_j^{-2},k_j^{-1},1,k_j,k_j^2,\cdots)^*
=(\cdots,l_j^2,l_j,1,l_j^{-1},l_j^{-2},\cdots)={\mf c}_{l_j}^T\bs\Theta,
\end{align*}
where $\bs\Theta$ is defined in \eqref{Theta}.
For $\mf C_1$ and $\mf C_2$, we have
\begin{align*}
	\mf C_1^*&=\Big(\sum_{j=1}^{M} \rho_1(k_j)\sigma_1(k_j) \mf c_{k_j}\mf c_{k_j}^T\Big)^* \\
	&=\Big(\sum_{j=1}^{M} \frac{\rho_1(k_j)}{k_j} k_j\mf c_{k_j}
         \mf c_{k_j}^T k_j\frac{\sigma_1(k_j)}{k_j}\Big)^* \\
	&=\bs\Lambda \sum_{j=1}^{M} \rho_2(l_j)\bs\Theta{\mf c}_{l_j}
         {\mf c}_{l_j}^T \bs\Theta \sigma_2(l_j)\bs\Lambda^T \\
	&=\bs\Lambda\bs\Theta{\mf C}_{2}\bs\Theta\bs\Lambda^T,
\end{align*}
and
\begin{align*}
	\mf C_2^*=\bs\Lambda\bs\Theta{\mf C}_{1}\bs\Theta\bs\Lambda^T.
\end{align*}
For $\mf U, \mf V$, $\overline{\mf U}$ and $\overline{\mf V}$, we have
\begin{align*}
	\mf U^*&=((\mf C_1^*)^{-1}-\bs\Omega\mf C_2^*\bs\Omega)^{-1} \\
	&=((\bs\Lambda\bs\Theta{\mf C_2}\bs\Theta\bs\Lambda^T)^{-1}
         -\bs\Omega\bs\Lambda\bs\Theta{\mf C_1}\bs\Theta\bs\Lambda^T\bs\Omega)^{-1} \\
	&=(\bs\Lambda\bs\Theta{\mf C_2}^{-1}\bs\Theta\bs\Lambda^T
         - \bs\Lambda \bs\Lambda^T\bs\Omega\bs\Lambda \bs\Theta
         {\mf C_1}\bs\Theta \bs\Lambda^T\bs\Omega\bs\Lambda \bs\Lambda^T)^{-1} \\
	&=\bs\Lambda\bs\Theta({\mf C_2}^{-1}
-{\bs\Omega}{\mf C_1}\bs\Omega)^{-1}\bs\Theta\bs\Lambda^T  \\
	&=\bs\Lambda\bs\Theta{\mf V}\bs\Theta\bs\Lambda^T,
\end{align*}
and
\begin{align*}
	\overline{\mf U} ^*&=-((\mf C_2^*\bs\Omega\mf C_1^*)^{-1}-\bs\Omega)^{-1} \\
	&=((\bs\Lambda\bs\Theta{\mf C_1}\bs\Theta(-\bs\Lambda^T\bs\Omega \bs\Lambda)\bs\Theta
       {\mf C_2}\bs\Theta\bs\Lambda^T)^{-1}
       -\bs\Lambda(-\bs\Lambda^T\bs\Omega\bs\Lambda)\bs\Lambda^T)^{-1} \\
	&=\bs\Lambda\bs\Theta(({\mf C_1}{\bs\Omega}{\mf C_2})^{-1}-{\bs\Omega})^{-1}\bs\Theta\bs\Lambda^T \\
	&=-\bs\Lambda\bs\Theta\overline{\mf V}\bs\Theta\bs\Lambda^T.
\end{align*}
Hence we have
\bsb
\begin{align}
	&(\mf U^*)_{i,j}=(\bs\Lambda\bs\Theta{\mf V}\bs\Theta\bs\Lambda^T)_{i,j}
       =(\bs\Theta{\mf V}\bs\Theta)_{i+1,j+1}=(\mf V)_{-i-1,-j-1}, \\
	&(\overline{\mf U}^*)_{i,j}=-(\bs\Lambda\bs\Theta\overline{\mf V}\bs\Theta\bs\Lambda^T)_{i,j}
       =-(\bs\Theta\overline{\mf V}\bs\Theta)_{i+1,j+1}=-(\overline{\mf V})_{-i-1,-j-1}.
\end{align}
\esb
We are led to
\begin{align*}
	\bs P^*=
	\begin{pmatrix}
		1-(\overline{\mf V}^*)_{0,-1} & -(\mf V^*)_{0,-1} \\
		-(\mf U^*)_{0,-1} & 1-(\overline{\mf U}^*)_{0,-1}
	\end{pmatrix}
	=
	\begin{pmatrix}
		1-(\overline{\mf V})_{0,-1} & -(\mf U)_{0,-1} \\
		-(\mf V)_{0,-1} & 1-(\overline{\mf U})_{0,-1}
	\end{pmatrix}
	=\bs P^T,
\end{align*}
which shows $\bs P$ is a Hermitian matrix, and therefore $\bs P$ satisfies the SU(2) gauge condition in
the space $\mathbb U_1$.

\subsection{Reduction to the SU(2) ASDYM equation in $\mb U_2$}\label{sec-4-3}

The reality condition in  $\mb U_2$ is fulfilled by the coordinate transformation:
\begin{equation}
	\tilde{z}\doteq-\mr ix_{n+1}=\xi_{n+1}, ~~~
	z\doteq-\mr ix_{m}=\xi_m, ~~~
	\tilde{w}\doteq-\mr ix_{m+1}=\xi_{m+1}, ~~~
	w\doteq-\mr ix_{n}=\xi_n,
\end{equation}
where $\xi_n,\xi_{n+1},\xi_m$ and $\xi_{m+1}$ are real.
With such a setting, the unreduced ASDYM equation \eqref{Yang-J-SDYM-2} reduces to
the ASDYM equation
\begin{equation}
	\partial_{\tilde z}(\bs P^{-1}\partial_{z}\bs P)-\partial_{\tilde w}(\bs P^{-1}\partial_{w}\bs P)=0.
\end{equation}

Following the settings in \eqref{measure} for the measures $\mu_1(l)$ and $\mu_2(l')$, we have the following
plane wave factors:
\begin{subequations}\label{PWF-U2}
\begin{align}
	&\rho_1(k_j)=\exp\big(\hat\zeta_{n,m}(k_j)+\lambda_1(k_j) \big),    ~~ &
      \rho_2(l_j)=\exp\big(-\hat\zeta_{n,m}(l_j)+\lambda_2(l_j) \big),\\
	&\sigma_1(k_j)=\exp\big(\hat\zeta_{n,m}(k_j)+\gamma_1(k_j) \big), ~~ &
      \sigma_2(l_j)=\exp\big(-\hat\zeta_{n,m}(l_j)+ \gamma_2(l_j)\big),
\end{align}
with
\begin{equation}
\hat\zeta_{n,m}(k_j)=\mr ik_j^m\xi_m+\mr ik_j^{m+1}\xi_{m+1}+\mr ik_j^n\xi_n+\mr ik_j^{n+1}\xi_{n+1},
\end{equation}
\end{subequations}
and the matrices $\mf C_1$ and $\mf C_2$ take the  form \eqref{C1C2}
but with the above plane wave factors.

To meet the gauge condition in Table \ref{Tab-2}, i.e. $\bs P^\dagger\bs P=\bs I_2$,
we introduce the constraints
\begin{equation}
	l_j=k_j^*, ~~~\exp(\lambda_2(l_j))=\exp(\lambda_1(k_j))^*, ~~~\exp(\mu_2(l_j))=-\exp(\mu_1(k_j))^*,
\end{equation}
which give rise to
$(\hat\zeta_{n,m}(k_j))^*=-\hat\zeta_{n,m}(l_j)$ and
\begin{equation}
	\rho_2(l_j)=(\rho_1(k_j))^*, ~~~\sigma_2(l_j)=-(\sigma_1(k_j))^*.
\end{equation}
Hence we have
\begin{align*}
	(\mf c_{k_j}^T)^* =(\cdots,k_j^{-2},k_j^{-1},1,k_j,k_j^2,\cdots)^*
=(\cdots,l_j^{-1},l_j^{-1},1,l_j,l_j^{2},\cdots)=\mf c_{l_j}^T,
\end{align*}
which, then, yields the following relation for $\mf C_1$ and $\mf C_2$:
\begin{align*}
	\mf C_1^* =\Big(\sum_{j=1}^{M} \rho_1(k_j)\sigma_1(k_j) \mf c_{k_j}\mf c_{k_j}^T\Big)^*
	 =-\sum_{j=1}^{M} \rho_2(l_j)\sigma_2(l_j) \mf c_{l_j}\mf c_{l_j}^T
	=-{\mf C}_{2}.
\end{align*}
Thus we have
\bsb\label{conjugate-U2}
\begin{align}
	\mf U^*&=((\mf C_1^*)^{-1}-\bs\Omega\mf C_2^*\bs\Omega)^{-1}
                  =-((\mf C_2)^{-1}-\bs\Omega\mf C_1\bs\Omega)^{-1}=-\mf V, \\
	\overline{\mf U}^*&=-((\mf C_2^*\bs\Omega\mf C_1^*)^{-1}-\bs\Omega)^{-1}
                 =-((\mf C_1\bs\Omega\mf C_2)^{-1}-\bs\Omega)^{-1}=\overline{\mf V}.
\end{align}
\esb

Next, considering $n=1$ in \eqref{recur-1}, we have
\begin{equation}
	\mk T^{(-1,0)}- \mk T^{(0,-1)}=\mk T^{(0,-1)} \mk T^{(0,-1)}.
\end{equation}
This gives rise to a relation
\begin{equation}\label{QP=I2}
	\bs W\bs P=\bs I_2,
\end{equation}
where $\bs P$ is defined as in \eqref{3.43} and
\[\bs W\doteq\bs I_2+ \mk T^{(-1,0)}.\]
Then, utilizing  \eqref{conjugate-U2} we can calculate the conjugate of $\bs P$:
\begin{align*}
	\bs P^*&=\begin{pmatrix}
		1-(\overline{\mf V}^*)_{0,-1} & -(\mf V^*)_{0,-1} \\
		-(\mf U^*)_{0,-1} & 1-(\overline{\mf U}^*)_{0,-1}
	\end{pmatrix}
	=\begin{pmatrix}
		1-(\overline{\mf U})_{0,-1} & (\mf U)_{0,-1} \\
		(\mf V)_{0,-1} & 1-(\overline{\mf V})_{0,-1}
	\end{pmatrix} \\
	&=\begin{pmatrix}
		1+(\overline{\mf V})_{-1,0} & (\mf U)_{-1,0} \\
		(\mf V)_{-1,0} & 1+(\overline{\mf U})_{-1,0}
	\end{pmatrix}
	=\bs I_2+\big(\mk T^{\mr T}\big)^{(-1,0)}=\bs W^T.
\end{align*}
This result, together with \eqref{QP=I2}, indicates
\begin{equation}
	\bs P^\dagger\bs P=\bs I_2.
\end{equation}
Hence, the solution $\bs P$ is unitary.

\section{Concluding remarks}\label{sec-5}

In this paper, we have established the DL scheme for the SU(2) ASDYM equation.
The main idea is to construct an unreduced ASDYM equation (i.e. \eqref{Yang-J-SDYM-2})
and then reduce it the ASDYM equation in various spaces.
We started from the linear integral equation set \eqref{int-eq} with separable measure \eqref{3.3}
and general plane wave factors \eqref{PWF-3.2}.
Then we introduced infinite-dimensional matrices $\{\mf U,\mf V,\overline{\mf U}, \overline{\mf V}\}$
as well as $\mk T$.
In this scheme we were able to derive their evolution relations and ``difference'' relations (see Proposition \ref{prop-4})
and differential recurrence relation (see Theorem \ref{Th-1}).
Then, we could extract from $\mk T$ the relations for the $2\times 2$ matrix $\mk T^{(i,j)}$
and derive the differential relation $(\partial_{n+1}\bs P)\bs P^{-1}=-\partial_n\bs Q$, i.e. \eqref{diff-recur}.
The compatibility $\partial_m(\partial_n\bs Q)=\partial_n(\partial_m\bs Q)$ gives rise to the
unreduced ASDYM equation \eqref{Yang-J-SDYM-2}.
The reduction includes choosing coordinate transformations to meet the reality conditions given in Table \ref{Tab-1}
and imposing constraints on $k_j$ and $l_j$ to meet the gauge conditions given in Table \ref{Tab-2}.
By reduction we have successfully got the SU(2) ASDYM equations in the Euclidean space $\mathbb{E}$
and two ultrahyperbolic spaces $\mathbb{U}_1$ and $\mathbb{U}_2$.
In this paper, we did not present more explicit forms for $\bs P$.
One can derive them along the lines of the technique used in \cite{ZZN-SAM-2012,NSZ-2023}.

Our DL scheme is actually based on the AKNS system.
Note that the (A)SDYM equation can also be formulated from the (matrix) KP system \cite{LSS-2023}.
It is possible to develop a different DL scheme using the connection with the KP system.
This will be one of topics to be investigated.

The DL scheme will provide a new platform to study the (A)SDYM equation.
One can extended the scheme to the SU($N$) SDYM equations.
Note that so far the solutions obtained from the Cauchy matrix approach \cite{LQYZ-SAPM-2022,LSS-2023}
do not provide  instantons (cf.\cite{Atiyah-1978,Corrigan-1977,Wilczek-1977}) with physical significance.
However, it is possible to introduce new plane wave factors in the DL scheme so that instantons can be derived.
In addition, the DL scheme has been extended from the usual soliton case to the elliptic soliton case \cite{NSZ-2023},
It would be interesting to develop an elliptic DL scheme for the SDYM equation.
Finally, we note that the DL scheme provides a powerful tool in integrable discretization.
It is expected to have a discrete version of the SDYM equation,
which might be more mathematically significant in discretizing the geometry structure of the SDYM equation.

\vskip 20pt
\subsection*{Acknowledgments}
This project is supported by the NSFC grant (No. 12271334, 12326428).
The authors are grateful to Prof. M. Hamanaka for discussion.

\vskip 20pt

\small{

}


\begin{thebibliography}{99}
	
\bibitem{Jimbo-1982} M. Jimbo, M.D. Kruskal, T. Miwa,
	 Painlev\'{e} test for the self-dual Yang-Mills equation,
	 Phys. Lett. A, 92 (1982) 59-60.
	
\bibitem{Ward-1984} R.S. Ward,
	The Painlev\'{e} property for the self-dual gauge-field equations,
	Phys. Lett. A, 102 (1984) 279-282.
	
\bibitem{Atiyah-1978} M.F. Atiyah, N.J. Hitchin, V.G. Drinfeld, Yu.I. Manin,
	Construction of instantons,
	Phys. Lett A, 65 (1978) 185-187.
	
\bibitem{Corrigan-1977} E. Corrigan, D.B. Fairlie,
	Scalar field theory and exact solutions to a classical SU(2) gauge theory,
	Phys. Lett. B, 67 (1977) 69-71.
	
\bibitem{Wilczek-1977} F. Wilczek,
	 Geometry and interactions of instantons,
     In: {Quark Confinement and Field Theory},
     Eds. D.R. Stump, D.H. Weingarten,
     Wiley, New York, 1977, pp211-219.
	
\bibitem{Yang-1977} C.N. Yang,
	Condition of self-duality for SU(2) gauge fields on Euclidean four-dimensional space,
	Phys. Rev. Lett., 38 (1977) 1377-1379.
	
	
\bibitem{BFNY-1978} Y. Brihaye, D.B. Fairlie, J. Nuyts, R.G. Yates,
	Properties of the self dual equations for an SU(n) gauge theory,
	J. Math. Phys., 19 (1978) 2528-2532.
	
\bibitem{Po-1980} K. Pohlmeyer,
	On the Lagrangian theory of anti-self-dual fields in four-dimensional Euclidean space,
	Commun. Math. Phys., 72 (1980) 37-47.
	
\bibitem{deVega-1988} H.J. de Vega,
	Non-linear multi-plane wave solutions of self-dual Yang-Mills theory,
	Commun. Math. Phys., 116 (1988) 659-674.
	
\bibitem{Getmanov-1990} B.S. Getmanov,
	$N$-monopole soliton-type solutions of the self-duality equations for an SU(2) gauge theory in Minkowski space-time,
	Phys. Lett. B, 244 (1990) 455-457.
	
\bibitem{Chau-1993} L.L. Chau, I. Yamanaka,
	Quantization of the self-dual Yang-Mills system:
    Exchange algebras and local quantum group in four-dimensional quantum field theories,
	Phys. Rev. Lett., 70 (1993) 1916-1919.
	
\bibitem{Chau-1994} L.L. Chau, J.C. Shaw, H.C. Yen,
	$N$-soliton-type solutions of SU(2) self-duality Yang-Mills equations in various spaces and
    their B\"{a}cklund transformations,
	J. Phys. A: Math. Gen., 27 (1994) 7131-7138.
		
\bibitem{Mason-2005} L. Mason,
	Global anti-self-dual Yang-Mills fields in split signature and their scattering,
	J. reine angew. Math., 597 (2006) 105-133.
	
\bibitem{Hamanaka-2020} M. Hamanaka, S.C. Huang,
	New soliton solutions of anti-self-dual Yang-Mills equations,
	J. High Energ. Phys., 2020 (2020) 101 (18pp).

\bibitem{Huang-2021} S.C. Huang,
	On Soliton Solutions of the Anti-Self-Dual Yang-Mills Equations from the Perspective of Integrable Systems
    (PhD thesis), Nagoya University, 2021,
	arXiv:2112.10702 [hep-th].
	
\bibitem{Hamanaka-2022} M. Hamanaka, S.C. Huang,
	Multi-soliton dynamics of anti-self-dual gauge fields,
	J. High Energ. Phys., 2022 (2022) 039 (19pp).
	
\bibitem{Hamanaka-2023} M. Hamanaka, S.C. Huang, H. Kanno,
	Solitons in open $N=2$ string theory,
	Prog. Theor. Exp. Phys., 2023 (2023) 043B03 (47pp).
	

	
\bibitem{Belavin-1978} A.A. Belavin, V.E. Zakharov,
	Yang-Mills equations as inverse scattering problem,
	Phys Lett. B, 73 (1978) 53-57.
	
\bibitem{Ueno-1982} K. Ueno, Y, Nakamura,
	Transformation theory for anti-self-dual equations and the Riemann-Hilbert problem,
	Phys Lett. B, 109 (1982) 273-278.
	
\bibitem{Takasaki-1984} K. Takasaki,
	A new approach to the self-dual Yang-Mills equations,
	Commun. Math. Phys., 94 (1984) 35-59.
	
\bibitem{Sasa-1998} N. Sasa, Y. Ohta, J. Matsukidaira,
	Bilinear form approach to the self-dual Yang-Mills equations and integrable systems in (2+1)-dimension,
	J. Phys. Soc. Japan, 67 (1998) 83-86.
	
\bibitem{Nimmo-2000} J.J.C. Nimmo, C.R. Gilson, Y.Ohta,
	Applications of Darboux transformations to the self-dual Yang-Mills equations.,
	Theor. Math. Phys., 122 (2000) 239-246.
	
\bibitem{LQYZ-SAPM-2022} S.S. Li, C.Z. Qu, X.X. Yi, D.J. Zhang,
	Cauchy matrix approach to the SU(2) self-dual Yang--Mills equation,
	Stud. Appl. Math., 148  (2022) 1703-1721.
	
\bibitem{Zhao-2018} S.L. Zhao,
	The Sylvester equation and integrable equations: The Ablowitz-Kaup-Newell-Segur system,
	Rep. Math. Phys., 82 (2018) 241-263.
	
\bibitem{LSS-2023} S.S. Li, C.Z. Qu, D.J. Zhang,
	Solutions to the SU($N$) self-dual Yang-Mills equation,
	Physica D, 453 (2023) 133828 (17pp).
	

	
\bibitem{Fokas-1981} A.S. Fokas, M.J. Ablowitz,
	Linearization of the Korteweg-de Vries and Painlev\'e II equations,
	Phys. Rev. Lett., 47 (1981) 1096-110.
	
\bibitem{Nijhoff-1983} F.W. Nijhoff, G.R.W. Quispel, J. Van der linden, H.W. Capel,
	On some linear integral equations generating solutions of nonlinear partial differential equations,
	Physica A, 119 (1983) 102-142.
	
\bibitem{Nijhoff-1983-2} F.W. Nijhoff, G.R.W. Quispel, J. Van der linden, H.W. Capel,
	The derivative nonlinear Schr\"{o}dinger equation and the massive thirring model,
	Phys. Lett. A, 93 (1983) 455-458.

\bibitem{83-NQC} F.W. Nijhoff, G.R.W. Quispel, H.W. Capel,
        Direct linearisation of difference-difference equations,
        Phys. Lett. A 97 (1983) 125-128.

\bibitem{QNCL-84} G.R.W. Quispel, F.W. Nijhoff, H.W. Capel, J. van ver Linden,
        Linear integral equations and nonlinear difference-difference equations,
        Physica A, 125 (1984) 344-380.

\bibitem{85-NCW} F.W.  Nijhoff, H.W. Capel, G.L. Wiersma,
        Integrable lattice systems in two and three dimensions,
        Ed. R. Martini, in: Geometric Aspects of the Einstein Equations and Integrable Systems,
        Lect. Not. Phys. 239 (1985) 263-302.

\bibitem{ZZN-SAM-2012} D.J. Zhang, S.L. Zhao, F.W. Nijhoff,
        Direct linearization of an extended lattice BSQ system,
        {Stud. Appl. Math.}, {129} (2012) 220-48.
	
\bibitem{FW-2017} W. Fu, F.W. Nijhoff,
	Direct linearizing transform for three-dimensional discrete integrable systems: the lattice AKP, BKP and CKP equations,
	Proc. R. Soc. A, 473 (2017) 20160915 (22pp).
	
\bibitem{FW-2020} W. Fu,
	Direct linearization approach to discrete integrable systems associated with $\mb Z_{\mathcal N}$ graded Lax pairs,
	Proc. R. Soc. A, 476 (2020) 20200036 (17pp).
	
	
\bibitem{FW-thesis} W. Fu,
	Direct Linearisation of Discrete and Continuous Integrable Systems: The KP Hierarchy and its Reductions (PhD thesis),
	University of Leeds, 2018.
	
\bibitem{FW-Nijhoff-2022} W. Fu, F.W. Nijhoff,
	On a coupled Kadomtsev-Petviashvili system associated with an elliptic curve,
	Stud. Appl. Math., 149 (2022) 1086-1122.

\bibitem{NSZ-2023} F.W. Nijhoff, Y.Y. Sun, D.J. Zhang,
     Elliptic solutions of Boussinesq type lattice equations and the elliptic N-th root of unity,
     Commun. Math. Phys., 399  (2023) 599-650.
	
\bibitem{MW-book} L.J. Mason, N.M.J. Woodhouse,
	Integrability, Self-Duality, and Twistor Theory,
	Oxford University Press, Oxford, New York, 1996.

	
	
	
\end{thebibliography}
\end{document}